\documentclass[a4paper,11pt]{article}

\usepackage{typearea}
\usepackage{setspace}

\usepackage{fullpage}

\usepackage{comment} 
\usepackage{bbm}
\usepackage{framed} 
\usepackage{url} 
\usepackage{complexity}
\usepackage{booktabs}
\usepackage{amsmath,amssymb}
\usepackage{float}

\usepackage{amsthm}
\usepackage{thmtools} 
\usepackage{thm-restate}
\usepackage{nicefrac}
\usepackage{calc}
\usepackage{enumerate}
\usepackage{enumitem}
\usepackage[usenames,dvipsnames]{xcolor}

\usepackage[ruled,vlined,linesnumbered]{algorithm2e}
\usepackage{forloop}

\usepackage{graphicx}
\usepackage[font=footnotesize,labelfont=bf]{subcaption}
\usepackage[font=footnotesize,labelfont=bf]{caption}
\usepackage[nobreak=true]{mdframed}
\usepackage{appendix}
\usepackage[noend]{algpseudocode}
\usepackage[colorinlistoftodos]{todonotes}
\usepackage{xr}
\usepackage{array}
\usepackage{xspace}
\definecolor{ForestGreen}{rgb}{0.1333,0.5451,0.1333}
\definecolor{DarkRed}{rgb}{0.8,0,0}
\definecolor{Red}{rgb}{1,0,0}
\usepackage[linktocpage=true,
pagebackref=true,colorlinks,
linkcolor=DarkRed,citecolor=ForestGreen,
bookmarks,bookmarksopen,bookmarksnumbered]{hyperref}

\usepackage[capitalise]{cleveref}
\crefrangelabelformat{enumi}{#3#1#4--#5#2#6}
\usepackage{parskip}
\usepackage{chngcntr}
\usepackage{mathtools,stackengine}
\usepackage{multirow}

\stackMath
\newcommand{\stackGeq}[1]{%
	\setbox0=\hbox{${}\mathrel{\stackon[-1pt]{\geq}{\scriptstyle\text{#1\strut}}}{}$}
	\xdef\tmpwd{\dimexpr\the\wd0\relax}
	\kern.5\tmpwd\mathclap{\box0}&\kern.5\tmpwd
}

\usepackage{nameref}

\usepackage{tikz}
\usetikzlibrary{patterns}

\allowdisplaybreaks

\DeclareMathOperator*{\expectation}{\mathbb{E}}

\newcommand{\expect}{\expectation\expectarg}
\DeclarePairedDelimiterX{\expectarg}[1]{[}{]}{%
	\ifnum\currentgrouptype=16 \else\begingroup\fi
	\activatebar#1
	\ifnum\currentgrouptype=16 \else\endgroup\fi
}

\DeclarePairedDelimiterX{\nicesetarg}[1]{\{}{\}}{%
	\ifnum\currentgrouptype=16 \else\begingroup\fi
	\activatebar#1
	\ifnum\currentgrouptype=16 \else\endgroup\fi
}

\newcommand{\innermid}{\nonscript\;\delimsize\vert\nonscript\;}
\newcommand{\activatebar}{%
	\begingroup\lccode`\~=`\|
	\lowercase{\endgroup\let~}\innermid 
	\mathcode`|=\string"8000
}

\newcommand\opt{\textsc{Opt}\xspace}

\newcommand\alg{\textsc{Alg}\xspace}

\counterwithin{equation}{section}

\usepackage{eqparbox}

\theoremstyle{plain}

\newtheorem{theorem}{Theorem}[section]

\newtheorem{claim}[theorem]{Claim}

\newlength{\continueindent}
\usepackage{etoolbox}
\makeatletter
\newcommand*{\ALG@customparshape}{\parshape 2 \leftmargin \linewidth \dimexpr\ALG@tlm+\continueindent\relax \dimexpr\linewidth+\leftmargin-\ALG@tlm-\continueindent\relax}
\apptocmd{\ALG@beginblock}{\ALG@customparshape}{}{\errmessage{failed to patch}}
\makeatother

\makeatletter
\def\thm@space@setup{%
	\thm@preskip=\parskip \thm@postskip=0pt
}
\makeatother

\usepackage{etoolbox}
\usepackage{tikz}
\usetikzlibrary{tikzmark}
\usetikzlibrary{calc}

\errorcontextlines\maxdimen

\newcommand{\ALGtikzmarkcolor}{black}
\newcommand{\ALGtikzmarkextraindent}{4pt}
\newcommand{\ALGtikzmarkverticaloffsetstart}{-.5ex}
\newcommand{\ALGtikzmarkverticaloffsetend}{-.5ex}
\makeatletter
\newcounter{ALG@tikzmark@tempcnta}

\newcommand\ALG@tikzmark@start{%
	\global\let\ALG@tikzmark@last\ALG@tikzmark@starttext%
	\expandafter\edef\csname ALG@tikzmark@\theALG@nested\endcsname{\theALG@tikzmark@tempcnta}%
	\tikzmark{ALG@tikzmark@start@\csname ALG@tikzmark@\theALG@nested\endcsname}%
	\addtocounter{ALG@tikzmark@tempcnta}{1}%
}

\def\ALG@tikzmark@starttext{start}
\newcommand\ALG@tikzmark@end{%
	\ifx\ALG@tikzmark@last\ALG@tikzmark@starttext
	\else
	\tikzmark{ALG@tikzmark@end@\csname ALG@tikzmark@\theALG@nested\endcsname}%
	\tikz[overlay,remember picture] \draw[\ALGtikzmarkcolor] let \p{S}=($(pic cs:ALG@tikzmark@start@\csname ALG@tikzmark@\theALG@nested\endcsname)+(\ALGtikzmarkextraindent,\ALGtikzmarkverticaloffsetstart)$), \p{E}=($(pic cs:ALG@tikzmark@end@\csname ALG@tikzmark@\theALG@nested\endcsname)+(\ALGtikzmarkextraindent,\ALGtikzmarkverticaloffsetend)$) in (\x{S},\y{S})--(\x{S},\y{E});%
	\fi
	\gdef\ALG@tikzmark@last{end}%
}

\apptocmd{\ALG@beginblock}{\ALG@tikzmark@start}{}{\errmessage{failed to patch}}
\pretocmd{\ALG@endblock}{\ALG@tikzmark@end}{}{\errmessage{failed to patch}}
\makeatother

\algblock[with]{With}{EndWith}
\algblockdefx[With]{With}{EndWith}%
[1]{\textbf{with} #1 \textbf{do}}%
{}

\makeatletter
\ifthenelse{\equal{\ALG@noend}{t}}%
{\algtext*{EndWith}}
{}%
\makeatother

\usepackage{geometry}
\usepackage{caption}
\usepackage{verbatim}
\usepackage{amsmath}
\usepackage{amsfonts}   
\usepackage{amssymb}    
\usepackage[utf8]{inputenc} 
\usepackage{lastpage}
\usepackage{fancyhdr}
\usepackage{graphicx}  
\usepackage{color}     
\usepackage{enumitem}
\usepackage{cite}      
\usepackage{hyperref}  
\usepackage{longtable}
\usepackage{listings}
\usepackage{bold-extra}
\usepackage{mathrsfs}
\usepackage[percent]{overpic}
\usepackage{amsthm}
\renewcommand{\thefootnote}{\fnsymbol{footnote}}

\date{\today}

\author{
Yossi Azar\footnotemark[1]
\quad
Itamar Biran\footnotemark[1]
\quad
Amos Fiat\footnotemark[1]
}
\title{Multi Choice Min Prophet}
\begin{document}

\maketitle
\footnotetext[1]{Blavatnik School of Computer Science, Tel Aviv University, Israel.
Emails: \texttt{azar@tau.ac.il}, \texttt{itamarbiran@mail.tau.ac.il}, \texttt{fiat@tau.ac.il}.}

\renewcommand{\thefootnote}{\arabic{footnote}}
\setcounter{footnote}{0}
\begin{abstract}
  The prophet inequality is a fundamental problem in optimal stopping theory. Given $n$ independent variables drawn from known distributions, a player observes values sequentially and must decide irrevocably whether to stop and accept the current value or continue. The goal is to select a single element while maximizing the ratio between the value chosen and that of the maximum value in the sequence. In this paper, we study the minimization counterpart, often termed the min prophet or cost prophet inequality. Unlike the maximization setting, where simple threshold algorithms achieve half of the prophet's value, the minimization setting is significantly harder, with an exponential lower bound even for i.i.d.\ variables.

 We study a multi-choice relaxation in which the algorithm may select multiple variables and gets to choose the best amongst them (the minimum amongst those selected). Our goal is to minimize the expected number of selections while achieving a constant competitive ratio.
For adversarial order, we show that a constant competitive ratio requires a nearly linear number of choices in expectation, ergo, $\Omega(n/\ln n)$.
In contrast, we show that for the prophet secretary model (random order) one can attain constant competitiveness while requiring only an exponentially smaller expected number of choices i.e. $O(\ln n)$.
We give a refined analysis and define $M$ to be the ratio of the minimum expected value of any single variable to the expected minimum value of all variables (the prophet's value) and present an algorithm that achieves a constant competitive ratio with $O(\min\{\ln \ln M, \ln n\})$ choices in expectation for the prophet secretary. We show that this is tight up to low order log factors even for the special case of the i.i.d. model. Specifically, the lower bound on the expected number of choices for any constant competitive algorithm is $\Omega(\min\{\ln\ln M/\ln \ln \ln M,\ln n/\ln\ln n\})$. 
We also show that if we insist on a deterministic bound on the number of choices then every constant competitive algorithm requires $n$ choices. This holds even in the i.i.d.\ setting and shows that to achieve a constant competitive algorithm there is an exponential gap between the lower bound on the deterministic number of choices and the upper bound on the expected number of choices.
 
 Finally, we consider a variant where both the algorithm and the adversary choose  $r$ values and pay their sum, this is the minimization multi unit version. We extend our techniques to the multi-unit variant for i.i.d.\ variables, achieving a constant competitive ratio with a small expected number of choices.

\end{abstract}
\section{Introduction}
Consider the following problem, there are $n$ rewards 
$X_1,\dots,X_n$ drawn independently from known distributions 
$F_1,\dots,F_n$. The player observes the values one by one and may stop at 
any point, collecting the last observed reward. The goal is to maximize the 
payoff. This problem is well known and is called the prophet inequality \cite{KrengelSucheston1977,KrengelSucheston1978}. The optimum offline solution (prophet’s expected payoff) is $\mathbb{E}[\max_j X_j]$. For general variables, the best online algorithms achieve at least half of the optimum prophet's value \cite{KrengelSucheston1977,KrengelSucheston1978,SamuelCahn84}. For random order (also called prophet secretary) the bound is better (between $0.688$ \cite{LT25} and $0.723$ \cite{FredM23}). For i.i.d. variables, the optimal online algorithm achieves $0.745$ \cite{HillK82,KERTZ198688,Cor2017} . 
This result has been extended in many directions, including multi-choice 
prophet inequalities \cite{AssafSamuelCahn2000}, multi-unit prophet settings \cite{Alaei2014}, and resource augmentation \cite{BrustleEtAl2024}. 

In this work, we study the minimization counterpart, sometimes called the 
cost prophet inequality~\cite{LivanosMehta2022}. Here the $X_i$’s represent costs, 
and the player aims to minimize the selected value. For intuition, one may think of a traveler deciding when to buy a flight 
ticket: prices arrive sequentially, and upon seeing one the traveler must 
decide immediately whether to purchase it or continue waiting. We study the three main variants for prophet inequalities: adversarial order, prophet secretary (random order arrival), and i.i.d. variables.

One of the key differences between the minimization and maximization problems is that in minimization the algorithm must choose a variable, otherwise it incurs a cost of $\infty$. Unlike the 
reward version, the minimization problem behaves very differently: while for
certain i.i.d.\ distributions one can obtain a constant competitive ratio for example distributions with monotone hazard rate\cite{LivanosMehta2022}, for
general distributions it has been shown that even under i.i.d.\ assumptions 
only an exponential competitive ratio is achievable~\cite{EsfandiariEtAl2015}. This striking 
gap motivates our investigation of what guaranties can be obtained without such distributional assumptions.
  
We consider the multi-choice prophet inequality, in which the algorithm can choose multiple items. We compare the expected cost of the algorithm on the minimum of its choices (which is the last choice\footnote{Without loss of generality, the algorithm’s cost can be taken as the value of
its last accepted choice, since selecting a larger value than already chosen never improves
the outcome.}) to the prophet's cost. The case of a single choice corresponds to the classic cost prophet inequality where $n$ corresponds to the offline version. 

Our goal in this paper is to study how many choices are needed to achieve a constant competitive ratio for the min prophet problem, defined as the expected cost of the algorithm (the expected minimum of the choices made) over the expected cost of the optimal algorithm. 
We  say that algorithm $\alg$ is  $c$-competitive if  \[\mathbb{E}(\mbox{\rm cost}(\alg)) \le c \cdot \mathbb{E}(\mbox{\rm cost}(\opt))= c\cdot\mathbb{E}[\min_j X_j].\] Since the choices made by the algorithm are adaptive, we view the number of choices as a random variable, and we are interested in its expectation.


To state our results, we use the parameter $M$ which captures the power of the minimum of $n$ variables compared with a single variable. Specifically, this parameter is
\[M=\min_j\expect{X_j}/\expect{\min_j X_j}.
\]
Intuitively, $M$ measures how much better the global minimum is, in expectation,
over a single draw. For many light-tailed distributions $M$ grows only
polynomially with $n$, while for heavy-tailed distributions $M$ can be much
larger. We illustrate this for concrete examples (uniform and truncated
heavy-tailed distributions) in Appendix~\ref{App: M examples}.

First, we show that for adversarial order, a large number of choices are required. Specifically, any constant competitive algorithm requires  $\Omega(\frac{n}{\ln n})$ choices in expectation. Surprisingly, for the random order (prophet secretary model), one can get an exponential improvement. Our main result is that a constant competitive ratio is possible with $O(\min\{\ln\ln M,\ln n\})$ choices in expectation, and this is the best possible (up to low order factors). Interestingly, the lower bound works even for i.i.d. variables (in which case the order is irrelevant).
Moreover, we show that if we insist on a deterministic upper bound on the number of choices, then, to achieve a constant competitive algorithm $n$ choices are needed, even for i.i.d. variables.
The most technically challenging problem is the lower bound on the expected number of choices.

\subsection{Our results}

We obtain several new results for the min prophet inequality under the multi-choice variant and the multi-unit variant. Recall that $n$ is the number of variables and $M$ is the ratio between the minimum of the expectation and the expectation of the minimum.
\paragraph*{Multi-choice setting} 
\begin{itemize}
\item For adversarial order, we establish that any constant competitive algorithm requires 
$\Omega(\frac{n}{\ln n})$ choices in expectation. See Section \ref{sec:non_iid}.
\item
For the prophet secretary model (variables that arrive at a random order), we present a simple threshold algorithm (Section \ref{sec:algorithm}) that achieves a 
constant competitive ratio using only 
$O\!\left(\min\{\ln\ln M, \ln n\}\right)$ choices in expectation. 
To complement this result, we show that even in the i.i.d.\ setting, any constant-competitive algorithm requires at least 
\begin{equation}\Omega\!\left(\min\left\{\tfrac{\ln\ln M}{\ln\ln\ln M}, \tfrac{\ln n}{\ln\ln n}\right\}\right)=\widetilde{\Omega}\left(\min\{\ln\ln M, \ln n\}\right)\nonumber\end{equation}
choices in expectation, 
establishing that our result is tight up to a low order factor. See Section \ref{sec:lowerbound_expected}.
\item We show that ``choices in expectation" above is required even for i.i.d. variables. We prove that no algorithm that deterministically makes less than  $n$ choices can be
constant-competitive (see Section \ref{sec:worstcase}). In combination with the upper bound above, this shows
an exponential gap between requiring a deterministic and an expected number of choices. 

\end{itemize}

These results illustrate the fundamental gap between maximization 
and minimization prophet inequalities. In particular, while the classical reward 
setting admits constant-competitive algorithms without multiple choices, the cost 
setting provably requires multiple choices, and any non-trivial result is possible only under the relaxation that the number of choices holds in expectation. Moreover, in the maximization version, the  competitive ratios for adversarial and random order settings differ by a constant factor. Conversely, to achieve any constant competitive ratio for the minimization problem, the two settings (adversarial vs. random order) are very different, and there is an exponential gap between the number of choices required in the two settings. 

\paragraph*{Multi-unit setting}
We generalize the multi choice problem to the {\sl multi unit min prophet problem}. In this setting, both the online algorithm and the adversary have to collect $r$ items and may replace previously chosen items with other choices. One may think of this setting as having a series of tasks that need to be done, and one must do $r$ of these tasks. The cost associated with the $r$ items eventually chosen is the sum of their values.

The algorithm may make multiple choices but always keeps the $r$ smallest cost items amongst them. As above, we are interested in the number of choices made in expectation and the competitive ratio attained. The multi unit problem reduces to the multi choice problem for $r=1$.

We design a threshold-based algorithm, closely related to that of the multi-choice setting, that achieves a constant competitive ratio for i.i.d.\ variables while using only a small expected number of choices:
If $r \le \ln M$, the algorithm achieves a constant competitive ratio with
    \[
    O\!\left(\min\!\left\{r\ln(\tfrac{\ln M}{r}),\;r\ln (\tfrac{n}{r})\right\}\right)
    \]
    choices in expectation.
If $r \ge \ln M$, the algorithm achieves a constant competitive ratio with only $O(r)$ choices in expectation. See Appendix~\ref{sec:multiunit_algorithm}. Moreover, for any $1\leq r\leq n-1$, we show that if the algorithm is restricted to making  $r$ selections in total, i.e., no extra choices, then no constant competitive ratio is possible (even for i.i.d.\ inputs). See Appendix~\ref{sec:multiunit_lowerbound}.

\subsection{Techniques}
For the upper bound, we employ a threshold-based algorithm, calibrated so that the expected number of relevant samples (below the threshold) is logarithmic in $M$. We analyze the expected number of choices using record statistics, relying on Jensen's inequality to bound the updates to $O(\min\{\ln \ln M, \ln n\})$.

For our lower bounds, we construct hard instances using heavy-tailed distributions. For the  adversarial order lower bound, we utilize $n$ distinct heavy-tailed continuous distributions to demonstrate that an algorithm must switch its selection with high probability at nearly every step to avoid incurring a large cost. For the lower bound on the expected number of choices  in the i.i.d. setting we divide the $n$ variables into intervals of increasing size that depend on $M$ and $n$, such that, in total, there are $\Omega(\min\{\ln\ln M/\ln \ln \ln M,\ln n/\ln\ln n\})$ intervals. We prove that any constant competitive algorithm must make a choice in each of these intervals with very high probability, thereby establishing the required lower bound. For the lower bound with deterministic number of choices, we construct a discrete distribution with super-exponentially growing values. We show that for any algorithm restricted to $n-1$ choices, there exists a specific realization sequence where the algorithm fails to select the minimum, leading to an unbounded competitive ratio.

\subsection{Organization}

The remainder of the paper is organized as follows. Section~\ref{sec:preliminaries} formalizes the multi-choice minimization prophet setting.
Section~\ref{sec:non_iid} shows that for adversarial order any constant-competitive algorithm must perform $\Omega({n}/{\ln n})$ choices in expectation. Section~\ref{sec:algorithm} presents our constant-competitive threshold algorithm and analyzes both its performance and the expected number of choices it makes. Section~\ref{sec:lowerbound_expected} complements this with a nearly tight lower bound on the expected number of choices required by any constant-competitive algorithm. Section~\ref{sec:worstcase} turns to deterministic choices, proving that any algorithm restricted to at most $n-1$ choices cannot achieve a constant competitive ratio.  
Some proof are deferred to the Appendices \ref{app: non_iid}, \ref{app: lowerbound_expected}, Appendices~\ref{sec:multiunit_algorithm} and~\ref{sec:multiunit_lowerbound} focus on the multi-unit version of the problem: 
Appendix~\ref{sec:multiunit_algorithm} presents an algorithm that is almost identical to the threshold-based algorithm of Section~\ref{sec:algorithm}, proving that it achieves a constant competitive ratio while using only a small expected number of choices, 
and Appendix~\ref{sec:multiunit_lowerbound} establishes an exponential lower bound, showing that when the algorithm is allowed only $r$ selections, no constant competitive ratio can be guaranteed. Finally, Appendix \ref{App: M examples} gives  values of $M$ for some specific distributions.

\subsection{Related work}

\paragraph*{Min prophet inequality}
The first study of the min prophet problem (also called the cost prophet inequality) was by Assaf et al.\ \cite{Assaf_Goldstein_Samuel-Cahn_2004}, 
who considered the case of two choices under certain distributions and compared the algorithm’s 
performance with both the one-choice version and the prophet’s value. 
In contrast to the maximization setting, where constant competitive ratios are achievable, 
their results already suggested that the minimization variant is significantly harder. 
After a long gap, Esfandiari et al.\ \cite{EsfandiariEtAl2015} established that even in the 
i.i.d.\ setting the competitive ratio can be exponential, highlighting a stark difference from 
the reward prophet inequality, where a ratio of 2 is attainable even for non i.i.d.\ distributions. 
This result shifted the research focus towards identifying structural assumptions—such as 
distribution classes or fractional cost models, under which constant competitive guarantees  may still be possible. 

Livanos and Mehta \cite{LivanosMehta2022} analyzed the optimal algorithm for i.i.d.\ variables 
and showed that constant guarantees can be obtained under additional assumptions on the 
distributions, specifically through properties of the hazard rate. 
More recently, Livanos \cite{Livanos2024} extended this line of work using extreme value theory. 
For distributions where the maxima and minima of a sample converge in distribution, they proved 
that for single-threshold algorithms the optimal competitive ratio in the i.i.d.\ setting is 
polylogarithmic in~$n$, and that in the $k$-multi-unit minimization prophet inequality 
there exist constant-competitive single-threshold algorithms whenever $k \geq \log n$. 

Another related variant is the prophet inequality with monomial cost functions, in which each 
random variable corresponds to a cost function and the player may purchase a fractional amount 
$x_i \in [0,1]$ from each, subject to the total mass summing to 1. 
Qin et al.\ \cite{Junjie24} analyzed this model and, under mild regularity conditions on the 
distributions, proved that one can obtain a competitive ratio approaching~1 asymptotically.
In our work, all proofs hold for arbitrary distributions.

\paragraph*{Max prophet inequality}
The classical prophet inequality was introduced by Krengel and Sucheston 
\cite{KrengelSucheston1977,KrengelSucheston1978}, who showed that a gambler observing 
$n$ independent random variables can guarantee at least half of the expected maximum. 
The optimal asymptotic competitive ratio for the i.i.d.\ prophet inequality is
\(\alpha \approx 0.745\), obtained by combining the Hill--Kertz upper bound \cite{HillK82} with the matching threshold-algorithm guarantee of Correa et al. \cite{Cor2017}. 

An intermediate variant is the \emph{prophet secretary} problem, where variables arrive 
in uniformly random order. Esfandiari et al.\ \cite{EsfandiariEtAl2015} designed an 
adaptive-threshold algorithm achieving a $(1-1/e)$-competitive ratio. Azar et al. and later Chen et al. \cite{Azar18,LT25} improved this bound to $0.688$.   Correa et al.\ \cite{Correa21} 
established that no algorithm can beat~$0.732$. Further progress has been made for multiple 
selections \cite{Arnosti23} and for broader feasibility constraints such as matroids 
\cite{Eshani2024,Marek20}. Yet another model, the free-order 
setting (where the algorithm can choose the order) was also studied in \cite{Abolhassani17}.

The most general problem is that of adversarial order. As for prophet secretary, many generalizations have
been considered. The classical $1/2$ guarantee extends to general downward-closed 
systems, including matroids \cite{KLEINBERG2019}. In the special case where the player 
may select up to $k$ values and is evaluated on their sum, Alaei \cite{Alaei2014} proved 
a $(1-\sqrt{1/(k+3)})$-competitive algorithm, later improved for small $k$ \cite{Chawala24} 
and shown to be tight for all $k$ \cite{Jiashou22}. A related model evaluates the player 
on the maximum of $k$ choices: Assaf and Samuel-Cahn \cite{AssafSamuelCahn2000} showed 
that in this setting the gambler can secure more than $\tfrac{k}{k+1}$ of the prophet’s value. 
 Ezra et al.\ \cite{Ezra18} introduced the 
$\ell$-out-of-$k$ framework, which interpolates between these models by allowing up to 
$k$ tentative selections while only the best $\ell$ are evaluated. In the special case 
$\ell=1$, they improved the bound of Assaf and Samuel-Cahn by giving an algorithm that 
achieves $1-\tfrac{3}{2}e^{-k/6}$ of the prophet’s value. This guarantee was later improved in \cite{Harb25}; related stronger guarantees for oracle-augmented variants were obtained in \cite{HarbPeled24}.
Another relevant direction is the prophet inequality with \emph{recourse} (or buyback), 
where the decision maker can discard previously selected elements to accept new ones, 
often subject to a cancellation cost. Recent work by Ekbatani et al.\ \cite{EkbataniNNV24} 
formalized this setting, showing that the ability to cancel decisions allows for 
improved competitive ratios in settings where standard irrevocable algorithms are limited. 
Finally, another approach is \emph{resource augmentation}, where the algorithm is granted 
multiple samples or additional copies of the instance. Brüstle et al.\ \cite{BrustleEtAl2024} 
proved that block-threshold algorithms require only 
$\Theta(\log\log(1/\varepsilon))$ extra copies to obtain a $(1-\varepsilon)$-approximation.
\section{Preliminaries}
\label{sec:preliminaries}
In this section, we will formalize the problems considered herein.
We introduce the parameter $M$, defined 
as \[M=\min_j\expect{X_j}/ \, \expect{\min_j X_j}.
\]
Note that for i.i.d. variables for all $j$ the values $\expect{X_j}$ are the same, so the minimum in the numerator becomes redundant.

We assume that $\min_j\expect{X_j} > 0$,
otherwise there exists $X_i$ s.t. $\Pr[X_i=0] = 1$ and in this case the algorithm is trivial.
Unless stated otherwise, we assume that all random variables 
$X_1,\dots,X_n$ are independent. Note that the algorithm must choose at least one variable otherwise the competitive ratio is infinity. In general the distributions $F_1,\dots,F_n$ 
may be arbitrary. We consider adversarial order, random arrival (prophet secretary), and the i.i.d. model. 

\paragraph*{Multi Choice}

We study a multi-choice variant of the prophet inequality. 
An online algorithm is given $n$ random variables 
$X_1, X_2, \dots, X_n$ drawn independently from distributions 
$F_1, F_2, \dots, F_n$, arriving sequentially. 
The algorithm is allowed to make multiple selections. 
Its cost is defined as the minimum chosen value. 
Formally, if the algorithm selects indices $i_1, \dots, i_k$, then its cost is
\[
\text{ALG} = \min \{ X_{i_1}, \dots, X_{i_k} \}.
\]
Obviously, no algorithm will ever choose an item of higher cost than a current option, so the minimal cost item will always be the last chosen. 

The benchmark is the \emph{prophet}, who observes all $n$ values in advance 
and pays
\[
\text{OPT} = \min_{1 \leq i \leq n} X_i.
\]
We compare the expected cost of the algorithm to the expected cost of the prophet 
i.e.,
\[
\sup_{F_1,\dots,F_n} 
\frac{\mathbb{E}[\text{ALG}]}{\mathbb{E}[\text{OPT}]}.
\]

Special cases of this model include:
\begin{itemize}
    \item single choice: the classic prophet inequality for minimization.
    \item $n$ choices: the offline optimum, where the algorithm is equivalent to the prophet.
\end{itemize}

Our goal is to determine how many choices are required (in expectation) for an algorithm to guarantee 
a constant competitive ratio. If we insist on a deterministic upper bound on the number of choices and seek a constant competitive ratio, then $n-1$ choices do not suffice, even in the i.i.d.\ setting. 

\paragraph*{Multi unit}

In the \emph{multi-unit minimization} variant, the algorithm must  acquire  \emph{exactly} $r$ items, where each selection represents purchasing one unit of the item.  
The total cost is the sum of the $r$ smallest values among the items it has selected.  Formally, if the algorithm makes $k$ selections at indices $i_1,\dots,i_k$ (with $k\ge r$), 
its cost is
\[
\text{ALG} = \sum_{j=1}^{r} X_{(j)}^{\text{ALG}},
\]
where $X_{(1)}^{\text{ALG}} \le X_{(2)}^{\text{ALG}} \le \dots \le X_{(k)}^{\text{ALG}}$ 
are the values of the items chosen by the algorithm.  
That is, the algorithm’s cost is determined by the smallest $r$ values among its $k$ selections.

The benchmark (the prophet) observes all values in advance and pays
\[
\text{OPT} = \sum_{j=1}^{r} X_{(j)},
\]
where $X_{(1)} \le X_{(2)} \le \dots \le X_{(n)}$ denote the order statistics of 
$\{X_1,\dots,X_n\}$.  
Thus, the prophet always takes the $r$ globally smallest items. The competitive ratio in this model is defined as
\[
\text{CR}_{\text{multi-unit}}(r)
=\sup_{F_1,\dots,F_n}\frac{\mathbb{E}[\text{ALG}]}{\mathbb{E}[\text{OPT}]}.
\]

This model captures settings in which the decision-maker must procure multiple units
(e.g., resources or tasks) and seeks to minimize the \emph{total} cost,
rather than a single value as in the multi-choice case.  

As shown later, if the algorithm is forced to choose \emph{exactly} $r$ values,
no algorithm can guarantee better than an \emph{exponential} competitive ratio. This impossibility motivates relaxing the model so that the algorithm may make any number of selections and instead aims to \emph{minimize the expected number of choices} required 
to achieve a constant competitive ratio.

In our algorithm, we assume that the distributions are continuous. This is not a restriction, as if the distribution contains atoms, then we use the standard technique of perturbing the atoms (or just make the algorithm randomized at that point of mass). Also, our lower bounds use distributions with atoms of probability mass, but the distributions can be perturbed slightly, giving an equivalent bound with continuous distributions.

\section{Adversarial order: almost linear lower bound on the expected number of choices}
We analyze the adversarial order and show that any constant-competitive algorithm for arbitrary independent variables must make at least an almost linear number of choices in \(n\).
\label{sec:non_iid}

\begin{theorem}
Any algorithm for arbitrary independent (but not identical) variables that is constant competitive requires $\Omega\left(\frac{n}{\ln n}\right)$ choices in expectation.
\end{theorem}
\paragraph*{Proof intuition}
The distributions are chosen so that, with high probability, the values decrease over time, meaning that each early variable offers an important opportunity to improve the algorithm's current choice. On the other hand, if the algorithm rejects one of these improving values, it is forced, with small but non-negligible probability, to keep a value much larger than the offline minimum. Although this bad event occurs with low probability, its contribution to the expectation is large enough that any algorithm that does not accept many of the early variables has an unbounded competitive ratio.

\begin{proof}

Define the following 
variables $X_1,..,X_n$ as follows. For $1\leq k\leq n-1$
\[
F_{X_k}(x)=
\begin{cases}
1 - \dfrac{1}{x^{{n^{2k-2n}}}}, & 1\le x < e^{n^{2n}}, \\[1mm]
1, & x\ge e^{n^{2n}},
\end{cases}
\]
and for $X_n$:
\[
F_{X_n}(x)=
\begin{cases}
1 - \dfrac{1}{x^{1-{\sum_{k=1}^{n-1}n^{2k-2n}}}}, & 1\le x < e^{n^{2n}}, \\[1mm]
1, & x\ge e^{n^{2n}}.
\end{cases}
\]
First, we claim the following
\begin{claim}
\label{claim:expectation non i.i.d}
The expectation of the minimum of \(n\) independent variables is $n^{2n} +1$
\end{claim}
\begin{proof}
See Appendix~\ref{prf: expectation non i.i.d}.
\end{proof}

We define $A_k \;=e^{\,n^{2n-2k-1}}$ and prove the following claim:
\begin{claim} \label{clm: X_k values}
For every \(1 \le k \le n-1\), Denote $Y_k$ as the minimum value for first $k$ variables. For every $k, $ $Y_k$ satisfies that
\[
Y_k\ge A_{k}= e^{n^{2n-2k-1}} 
\]
with probability of at least $1-\frac{2}{n}$ for large enough $n$.
\end{claim}
\begin{proof}
See appendix ~\ref{prf: X_k values}
\end{proof}

Next, for each $k$, define
$
Z_k = \min\{X_{k+1},\dots,X_n\}
$ 
which is the minimum of the remaining variables, and
$
C_k = \mathbb{E}\!\left[\min\!\left(A_{k-1},Z_k\right)\right]$, 
which is a lower bound on the algorithm's value if we do not choose $x_k$.
\begin{claim} \label{clm: lower_bound_pref}
For every $1\le k\le n-1 $ it holds that:
\[
C_k = \mathbb{E}\!\left[\min(A_{k-1}, Z_k)\right]\ge \frac{1}{3}\cdot n^{2n-2k-2}\cdot e^n.
\]

\begin{proof}
    \begin{eqnarray*}
  C_k=\mathbb{E}\left[\min\left(A_{k-1},Z_k\right)\right] &=& 
    1+\int_{1}^{A_{k-1}}\frac{1}{x^{1-\sum_{i=1}^{k}\frac{n^{2i}}{n^{2n}}}} dx\\
    &=& 1+ \frac{1}{\sum_{i=1}^{k}\frac{n^{2i}}{n^{2n}}}\cdot \left(A_{k-1}^{\sum_{i=1}^{k}\frac{n^{2i}}{n^{2n}}}-1\right)\\
    &\ge&  \frac{1}{\sum_{i=1}^{k}\frac{n^{2i}}{n^{2n}}}\cdot\left(e^{\frac{n^{2n}}{n^{2k-1}}{\sum_{i=1}^{k}\frac{n^{2i}}{n^{2n}}}}-1\right)\\
    &\geq & \frac{n^{2n}}{\sum_{i=1}^{k}n^{2i}}\cdot \left(e^{\frac{n^{2k}}{n^{2k-1}}}-1\right)\\
    &\geq & \frac{n^{2n}}{2n^{2k+2}}\cdot \left(e^n-1\right)\\
    &\geq & \frac{n^{2n}}{3n^{2k+2}}\cdot e^n = \frac{1}{3}\cdot n^{2n-2k-2}\cdot e^n \ ,
\end{eqnarray*}
where the second inequality holds because we decreased the exponent, and the next inequality holds because $n^{2k+2}$ dominates the sum.  

\end{proof}
\end{claim}
\begin{claim}\label{clm:choose_prefix}
Let $m= \frac{n}{4\ln n}$. Any constant-competitive algorithm must make $m$ choices among
$\{X_1,\dots,X_m\}$ with probability at least $\frac12$.
\end{claim}
\begin{proof}

Suppose there exists a constant competitive algorithm that, with probability of at least $\frac{1}{2}$, makes fewer than $m$ choices among $\{X_1,\dots,X_m\}$.  For each $j\in\{1,\dots,m\}$, let $G_j$ be the event that the algorithm does \emph{not} choose $X_j$.  Let $E$ be the event that the algorithm makes fewer than $m$ choices among $\{X_1,\dots,X_{m}\}$.  By assumption, $\Pr(E)\ge\frac12$.  On the event $E$ the algorithm must skip at least one index in $\{1,\dots,m\}$, hence by using the union bound we get 
\[
\sum_{j=1}^m \Pr(G_j)\;\ge\;\Pr\!\left(\bigcup_{j=1}^m G_j\right)\ge \Pr(E)\ge \frac12,
\]

and therefore there exists $k\in\{1,\dots,m\}$ such that
\[
\Pr(G_k)\;\ge\;\frac{1}{2m}=\frac{2\ln n}{n}.
\]

Hence, there exists a $k$ such that the algorithm didn't choose a variable with a probability of at least $\frac{2\ln n}{ n}$. Denote this variable as $X_{k}$. When the algorithm skipped the value of $X_k$, its decision was based only on the first $k-1$ variables. Hence,  we can assume it is holding a previously selected value—specifically, one of \( \{X_1, \dots, X_{k-1}\} \). We can assume using claim \ref{clm: X_k values} that with probability of at least 
\[\Pr[\min\{X_1,\dots,X_{k-1}\}\ge A_{k-1}] =\Pr [Y_{k-1}\ge A_{k-1}] \ge 1-\frac{2}{n},\] 
that the values of the variables \( X_1, \dots, X_{k-1} \) are at least \( A_{k-1} \); it follows that the algorithm's current value is at least \( A_{k-1} \). Consequently, the expected performance of the algorithm is lower bounded by the minimum of its current value and the minimum of the remaining unseen variables. Since it rejected \( X_{k} \), we can assume it continues with the same value when it observes $X_{k+1}$. Therefore, its expected cost is at least $C_k=\mathbb{E}\!\left[\min(A_{k-1}, Z_k)\right]$.
Hence by claim \ref{clm: lower_bound_pref} its expected cost is at least:
\[
C_k\ge\frac{1}{3}\cdot n^{2n-2k-2}\cdot e^n \ .
\]

Denote the event that $\min\{X_1,\dots,X_{k-1}\}\ge A_{k-1}$ as $S_k$. We can bound $\Pr(S_k \cap G_k)$ by using the fact that
\[
\Pr(S_k)= \Pr(S_k\cap G_k)+\Pr (S_k\cap G_k^c)\le \Pr(S_k\cap G_k)+\Pr(G_k^c).
\]
Hence,
\[
\Pr(S_k\cap G_k)\ge \Pr(S_k)-\Pr(G_k^c)\ge \left(1-\frac{2}{n}\right)-\left(1-\frac{2\ln n}{n}\right) =\frac{2\ln n}{n}-\frac{2}{n} \geq \frac{1}{n}.
\]
Where the last inequality holds for large enough $n$.
Therefore, we can bound the expectation of such an algorithm by at least
\[
\mathbb{E}[ALG]\ge \Pr(S_k\cap G_k)\cdot C_k\ge \frac{1}{n}\cdot \frac{1}{3}n^{2n-2k-2}\cdot e^n=\frac{1}{3}\cdot n^{2n-2k-3}\cdot e^n.
\]
Since, $\mathbb{E}[OPT]= 1+n^{2n}$ then the competitive ratio of the algorithm is
\[
\frac{\mathbb{E}[ALG]}{\mathbb{E}[OPT]}\ge \frac{\frac{1}{3}\cdot n^{2n-2k-3}\cdot e^n}{2\cdot n^{2n}} =\frac{1}{6}\cdot \frac{e^n}{n^{2k+3}}\ge \frac{1}{6} \frac{e^n}{n^{\frac{n}{2\ln n}+3}}= \frac{1}{6}\cdot \frac{e^n}{e^{n/2+3\ln n}}\ge e^{n/3}.
\]
Note that for large enough $n$ we have $e^{n/3}\ge n^3$. Hence, if an algorithm skips the variable $X_k$ with probability of more than $\frac{2\ln n}{n}$ then its competitive ratio is at least  $e^{n/3}$
which is not constant competitive. Hence, 
any algorithm must choose all of  $\{X_1,\dots,X_m\}$ with probability of at least $\frac{1}{2}$.
\end{proof}

To complete the proof, we note that in total each constant competitive algorithm for large enough $n$ makes at least
\[
\mathbb{E}[Choices]\ge \frac{1}{2}\cdot \frac{n}{4\ln n}=\frac{n}{8\ln n}
\]
choices in expectation, which is $\Omega(\frac{n}{\ln n})$ as required.  
\end{proof}
\section{Constant competitive algorithm with a small expected number of choices}
\label{sec:algorithm}
In this section, we present our constant-competitive algorithm and analyze both its competitive ratio and the expected number of choices it makes. The algorithm is threshold based. It establishes a threshold and only accepts values below it or the variable that has the minimum expectation. To describe the algorithm we assume that the variables have a continuous distribution.  We now present the algorithm: denote the first variable with the lowest expectation as $X_{i_0}$ and its index as $i_0$. 
\begin{center}
\fbox{\parbox{0.93\textwidth}{
\textbf{Algorithm 1}
\begin{enumerate}
    \item Find a threshold $t$ (e.g., by binary search) such that $\sum_i p_i = \min\{\ln M+1,n$\} where $p_i = \Pr(x_i \leq t)$.
    \item Accept the first item with a value at most t, \emph{or} $X_{i_0}$
    \item Thereafter, switch only if a newly observed value is strictly smaller than the value currently held and less than $t$.
\end{enumerate}
}}
\end{center}

We consider the random order model (the prophet secretary problem).
\begin{theorem}
For arbitrary variables, in random order, the above algorithm is constant-competitive and makes  $O(\min\{\ln\ln M,\ln n\})$ choices in expectation.
\end{theorem}

We next analyze the algorithm's performance. The analysis proceeds in two parts: we first bound its competitive ratio and then the expected number of choices it makes, where the expectation is taken over the randomness of the input. For ease of exposition we assume $\mathbb{E}[min\{X_1,...X_n\}]=1$ which implies that $\mathbb{E}[X_{i_0}]=M$.

\subsection{Competitive ratio and number of choices analysis}
We analyze the algorithm separately for the competitive ratio and for the expected number of choices.

\subsubsection{Competitive ratio analysis}
First, observe that if $\ln M+1\geq n$ then $p_i = 1$ for every $i$. Hence, the algorithm always achieves the minimum and is 1-competitive. 
If $\ln M +1<n$ then 
the only situation in which the algorithm may fail to achieve the minimum is when all samples except for the sample $X_{i_0}$ are above $t$. Since $\ln M +1 <n $ then 
$\sum_{i=1}^np_i =\ln M +1$ which means that $\sum_{i=1}^{n}p_i-p_{i_0} \geq \ln M$, and thus the probability that all variables except $X_{i_0}$ are above $t$ is:
\[
\Pr\!\Bigl(\forall i \in [n]\setminus\{i_0\} : X_i > t\Bigr)
= \prod_{i \neq i_0} (1-p_i) 
\;\le\; \Bigl(1-\tfrac{\ln M}{n-1}\Bigr)^{n-1} 
\;\le\; e^{-\ln M} 
= \tfrac{1}{M},
\]
where the first inequality is the mean inequality and the second one is $(1-\alpha)^n\le e^{-n\alpha}$.
In this case, the algorithm's expected cost is $M$ since it always chooses the variable $X_{i_0}$.

If there is a sample below $t$, then the algorithm achieves the minimum and is therefore $1$-competitive.
Combining these cases, the expected performance of the algorithm is
\begin{eqnarray*}
\mathbb{E}[\text{ALG}] \; &\le\; & 1 \cdot E[\min \{X_1,...,X_n\}]|\min\{X_1,...,X_n\}\leq t) \;+\; \tfrac{1}{M}\cdot E[X_{i_0}] \\
 & \:\leq\: &
E[\min\{X_1,...,X_n\}] + \frac{1}{M} \cdot M \;=\; 2.
\end{eqnarray*}
Hence, the competitive ratio of the algorithm is at most $2$.

\subsubsection{Number of Choices Analysis}
Denote $Z_1 = \text{number of samples below t,  }Z_2=\text{number of choices overall}$.
Then we have $\mathbb{E}[Z_1]=\min\{\ln M+1,n\}$.
As the number of expected new minima in a random permutation of length $n$ is $\sum_{i=1}^n1/i =H_n \le \ln (n+1) +1$. We get that $\mathbb{E}[Z_2|Z_1]\le \ln (1+Z_1) +2$ since, there is 1 choice that is being made regardless of the number of samples.
Then by the law of total expectation we have that 
\[
\mathbb{E}[Z_2]=\mathbb{E}[\mathbb{E}[Z_2|Z_1]]\leq\mathbb{E}[\ln (1+Z_1)+2] \ .
\]
By Jensen's inequality, since $\ln $ is a concave function, we get that $\mathbb{E}[\ln (Z_1+1)+2]\le \ln (\mathbb{E}[Z_1+1])+2$, and hence
\[
\mathbb{E}[Z_2]\le\mathbb{E}[\ln (Z_1+1)+2] \le \ln(\mathbb{E}[Z_1+1])+2=O\left(\min \{\ln \ln M,\ln n\}\right).
\]
As required.

\section{Lower bound on the expected number of choices}
\label{sec:lowerbound_expected}
We next establish a nearly tight lower bound, showing that our upper bound on the expected number of choices is optimal up to lower-order terms. The lower bound is actually proved even for the i.i.d. and hence obviously holds for the prophet secretary. The proof relies on constructing a hard i.i.d.\ instance that forces any algorithm to make sufficiently many selections (even in expectation) to remain constant competitive.

\begin{theorem}\label{thm:lower bound}
Any constant competitive algorithm for arbitrary i.i.d. variables requires $\Omega\left(\min\{\frac{\ln\ n}{\ln\ln n}, \frac{\ln \ln M}{\ln \ln \ln M}\}\right)$ choices in expectation.
\end{theorem}
\begin{proof}
We will show a lower bound of $\Omega\left(\frac{\ln \ln \beta}{\ln \ln \ln \beta}\right)$ on the number of choices for any parameter $\beta\leq e^n$ where $M=\Theta\left( \frac{\beta}{\ln \beta}\right)$. This yields the required bound for $M =O\left(\frac{e^n}{n}\right)$. For $M=\Omega\left(\frac{e^n}{n}\right)$ we use $\beta=e^n$ and get a lower bound of $\Omega\left(\frac{\ln n}{\ln \ln n}\right)$.

\paragraph*{Proof overview}The proof proceeds in four steps.
\begin{itemize}
  \item We construct a heavy-tailed i.i.d.\ distribution and show that
  $M=\Theta(\beta/\ln\beta)$.
  \item We partition the sequence into intervals whose sizes grow
  doubly logarithmically.
  \item We show that skipping any such interval significantly increases the algorithm's expected cost.
  \item We conclude that any constant-competitive algorithm must make at
  least one choice in each interval with very high probability.
\end{itemize}

First, we define a random variable \(X\) that has the following cumulative distribution function (CDF):
\[
F(x)=
\begin{cases}
1 - \dfrac{1}{x^{1/n}}, & 1\le x < \beta, \\[1mm]
1, & x\ge \beta,
\end{cases}
\]
where $9 \le \beta\leq e^n$.
The probability density function (PDF) is
\[
f(x) = \frac{d}{dx}F(x) = \frac{1}{n} x^{-1 - \frac{1}{n}}\,.
\]

There is also a point mass at \( x = \beta \):
\[
P(X =\beta) = 1 - \lim_{x \to (\beta)^-} F(x) = \frac{1}{\beta^{1/n}}.
\]

\begin{claim}
\label{clm: expectation i.i.d}
The expectation of a single variable \( X \)  is  between $\frac{\beta}{4}$ and $2\beta$.
\end{claim}
\begin{proof}
\[
\mathbb{E}[X] = \int_{1}^{\beta} x\,f(x)\,dx +\beta\cdot \frac{1}{\beta^{1/n}} =\frac{1}{n} \int_1^{\beta} x^{-\frac{1}{n}}\,dx + \beta^{1-\frac{1}{n}}.
\]
We compute the integral:
\[
\int_1^{\beta} x^{-\frac{1}{n}}\,dx = \left[ \frac{x^{1 - \frac{1}{n}}}{1 - \frac{1}{n}} \right]_1^{\beta}
= \frac{\beta^{1 - \frac{1}{n}} - 1}{1 - \frac{1}{n}}\ .
\]


Hence, the expectation becomes:
\[
\mathbb{E}[X] = \frac{1}{n} \cdot \frac{\beta^{1-\frac{1}{n}} - 1}{1 - \frac{1}{n}} + \beta^{1-\frac{1}{n}}
= \frac{1}{n} \cdot \frac{\beta^{1-\frac{1}{n}} - 1}{\frac{n - 1}{n}} + \beta^{1 - \frac{1}{n}}
= \frac{\beta^{1 - \frac{1}{n}} - 1}{n - 1} + \beta^{1-\frac{1}{n}}=\frac{n\beta^{1-\frac{1}{n}}-1}{n-1}.
\]
Hence, 
\[
\mathbb{E}[X] =\frac{n\beta^{1-\frac{1}{n}}-1}{n-1}\leq 2\beta,
\]

and  
\[
\mathbb{E}[X]=\frac{n\beta^{1-\frac{1}{n}}-1}{n-1}\ge \frac{\beta}{e}-1\ge \frac{\beta}{4},
\]
where $9\le \beta \le e^n$,
ergo
$\frac{\beta}{4}\le\mathbb{E}[X]\le2\beta$ as required.
\end{proof}

\begin{claim}
\label{clm: expectation min i.i.d}
The expectation of the minimum of \(n\) independent variables is $\ln \beta + 1$.
\end{claim}
\begin{proof}
See appendix ~\ref{prf:expectation min i.i.d}
\end{proof}
From the above two claims it follows that  
$
M=\frac{\mathbb{E}[X_i]}{\mathbb{E}[\min\{X_1,...,X_n\}]}=\Theta\left( \frac{\beta}{\ln \beta}\right)\ .
$

Next define for $1 \le k \le \frac{\ln\ln  \beta}{2\ln \ln \ln  \beta}-1$ the prefix size
$ n_k=\frac{n(\ln\ln \beta)^{2k}}{\ln \beta}$ .

We now prove the following:
\begin{claim} \label{clm:Y_k}
The minimum value encountered $Y_k$ in the first $n_k$ variables for $1\leq k\leq \frac{\ln\ln  \beta}{2\ln \ln \ln  \beta}-1$ satisfies $Y_k \ge e^{\frac{\ln \beta}{(\ln \ln \beta)^{2k+1}}}$  with probability of more than $1-\frac{1}{\ln \ln \beta}$.
\end{claim}
\begin{proof}

Denote $ A_k=e^{\frac{\ln \beta}{(\ln \ln \beta)^{2k+1}}} $. The CDF of $Y_{k}$ is 
$$F_{Y_k}(x) =1-\frac{1}{x^{\frac{n(\ln \ln \beta)^{2k}}{n\ln \beta}}}= 1-\frac{1}{x^{\frac{(\ln  \ln \beta)^{2k}}{\ln \beta}}} \ .$$
So the probability that $Y_{k}$ is bigger than $A_k$  is
$
\Pr[A_k \le Y_{k}] = 1-F_{Y_k}(A_k).
$
Bounding $F_{Y_k}(A_k)$:
\[
F_{Y_k}(A_k) = 1-\frac{1}{A_k^{\frac{(\ln  \ln \beta)^{2k}}{\ln \beta}}} = 1-\frac{1}{e^{\frac{\ln \beta}{(\ln \ln \beta)^{2k+1}}\cdot{\frac{(\ln  \ln \beta)^{2k}}{\ln \beta}}}} = 1- e^{-\frac{1}{\ln \ln \beta}} \leq 1-1+\frac{1}{\ln \ln \beta} = \frac{1}{\ln \ln \beta}\ .
\]
where the inequality implied by $e^{-a} \ge 1 - a$.
Hence,
$
\Pr[A_k \le Y_{k}] = 1-F_{Y_k}(A_k)\ge 1-\frac{1}{\ln \ln \beta} \ .
$
\end{proof}
Next, for all $ 1 \le k \le \frac{\ln\ln  \beta}{2\ln \ln \ln  \beta}-1$, define
$ Z_k = \min\{X_{n_k+1},\dots,X_n\},$
which is the expected minimum of the remaining variables, and 
$C_k = \mathbb{E}\!\left[\min\!\left(A_{k-1},Z_k\right)\right]$
which is a lower bound on the value of the algorithm if it skips all variables in indices between $n_{k-1}$ to $n_k$.

\begin{claim}\label{clm:Ck}
For every $1\le k\le T$, we have
$
C_k \;\ge\; \frac{\ln^2\beta}{2(\ln\ln\beta)^{2k}}.
$
\end{claim}
\begin{proof}
    \begin{eqnarray*}
C_k 
& = & \int_0^1 1\,dx \;+\; \int_1^{e^{\frac{\ln \beta}{(\ln \ln \beta)^{2k-1}}}}
  \frac{1}{x^{\frac{n-\frac{n(\ln \ln \beta)^{2k}}{\ln \beta}}{n}}}\,dx \\
& = & \int_0^1 1\,dx \;+\; \int_1^{e^{\frac{\ln \beta}{(\ln \ln \beta)^{2k-1}}}}
  \frac{1}{x^{1 - \frac{(\ln \ln \beta)^{2k}}{\ln \beta}}}\,dx \\
& = & 1 \;+\; \int_1^{e^{\frac{\ln \beta}{(\ln \ln \beta)^{2k-1}}}}
   x^{-1+\frac{(\ln \ln \beta)^{2k}}{\ln \beta}}\,dx \\
& = & 1 \;+\; \left[\frac{x^{\frac{(\ln \ln \beta)^{2k}}{\ln \beta}}}{\tfrac{(\ln \ln \beta)^{2k}}{\ln \beta}}\right]_{x=1}^{x=e^{\frac{\ln \beta}{(\ln \ln \beta)^{2k-1}}}} \\
& = & 1 \;+\; \frac{\ln \beta}{(\ln \ln \beta)^{2k}}\left(e^{\tfrac{\ln \beta}{(\ln \ln \beta)^{2k-1}}\cdot \tfrac{(\ln \ln \beta)^{2k}}{\ln \beta}}- 1\right)\\
& = & 1 + \frac{\ln \beta}{(\ln \ln \beta)^{2k}}\,\left(e^{\ln \ln \beta} - 1\right).
\end{eqnarray*}

Hence,
\begin{equation*}
C_k 
= 1 + \frac{\ln \beta}{(\ln \ln \beta)^{2k}}\,\left(e^{\ln \ln \beta}-1\right) 
= 1 + \frac{\ln \beta}{(\ln \ln \beta)^{2k}}\,\left(\ln \beta - 1\right)
\geq \frac{\ln^2\beta}{2(\ln\ln \beta)^{2k}},
\end{equation*}
for all $\beta \ge 5$.
\end{proof}
Now we are ready to complete our proof (with high probability):
\begin{claim} \label{clm:num of choices}
Let $T= \frac{\ln\ln\beta}{2\ln\ln\ln\beta}-1$. Any constant competitive algorithm must make at least $T$ choices with probability higher than $1-\frac{1}{\ln \ln \ln \beta}$ in the first $n_{T}$ variables.

\end{claim}
\begin{proof}

\begin{figure}[H]
  \centering
  \includegraphics[width=0.4\linewidth]{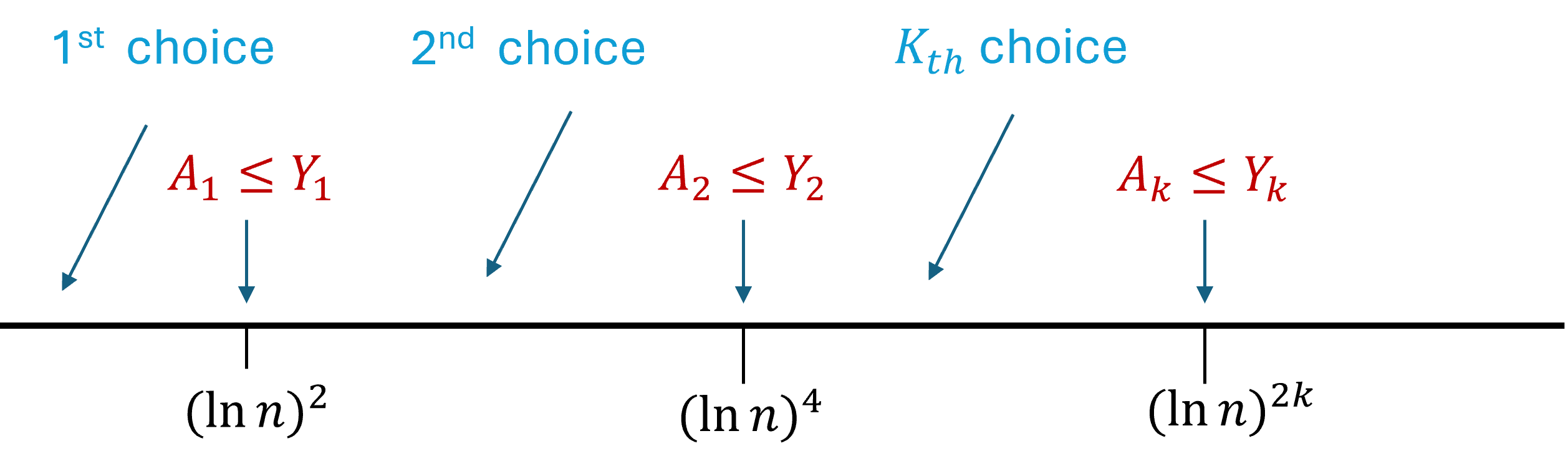}
  \caption{intervals for $\beta =e^n$}
  \label{fig:intervals}
\end{figure}

Assume by contradiction that there exists an algorithm which makes fewer than $T$ choices in the first $n_{T}$ variables with probability of more than $\frac{1}{\ln \ln \ln \beta }$ and that its competitive ratio is constant. By looking at the disjoint intervals $[1,n_1],[n_1+1,n_2],...,[n_{T-1}+1,n_{T}]$.

Let $E$ be the event that the algorithm makes fewer than $T$ choices in the first $n_T$ variables. On event $E$ it must skip at least one of the $T$ intervals, so by averaging, there exists an interval $k$ with
\[
\Pr(\text{skip interval }k)\;\ge\;\Pr(E)\cdot\frac{1}{T}\;\ge\;\frac{1}{\ln\ln\ln\beta}\cdot\frac{2\ln\ln\ln\beta}{\ln\ln\beta}\;=\;\frac{2}{\ln\ln\beta}.
\]

Denote the index of the interval it skipped on with the highest probability as $k$. Hence, the probability that the algorithm skipped on it is bigger than $\frac{2}{\ln \ln \beta}$.
When the algorithm skipped a value in interval $k$,  its decision was based only on the first $n_k$ variables. Hence, assuming $A_{k-1} \le Y_{k-1}$, all values revealed among the first $n_{k-1}$ variables are at least $A_{k-1}$; therefore any value chosen by the algorithm before interval $k$ is at least $A_{k-1}$, and assuming the algorithm makes no choice in interval $k$,  its expectation in this case is at least $\mathbb{E}[\min(A_{k-1},Z_k)]=C_k$ and by Claim~\ref{clm:Ck} we get that:
\[
\mathbb{E}[\min(A_{k-1},Z_k)]=C_k\ge \frac{\ln^2\beta}{2(\ln\ln \beta)^{2k}}.
\]
Denote the event that the algorithm skipped interval $k$ as $G$ and the event that $Y_{k-1}\ge A_{k-1}$ as $S_k$. We can bound $\Pr(S_k\cap G)$ by using the fact that 
\[
\Pr(S_k)= \Pr(S_k\cap G)+\Pr (S_k\cap G^c)\le \Pr(S_k\cap G)+\Pr(G^c).
\]
Hence,
\[
\Pr(S_k \cap G)\ge \Pr(S_k)-\Pr(G^c)\ge 1-\frac{1}{\ln \ln \beta}-\left(1-\frac{2}{\ln \ln \beta}\right)=\frac{1}{\ln \ln \beta}.
\]
Therefore, we can bound the expectation of the algorithm by:
\[
\mathbb{E}[ALG]\ge \Pr(S_k \cap G)\cdot C_k \geq \frac{1}{\ln \ln \beta}\cdot\frac{\ln^2\beta}{2(\ln\ln \beta)^{2k}}=\frac{\ln^2\beta }{2(\ln \ln \beta)^{2k+1}}.\]

Hence, the competitive ratio of such an algorithm is
\begin{eqnarray*}
\frac{\mathbb{E}[ALG]}{\mathbb{E}[OPT]}
&\ge&
\frac{\frac{\ln ^2\beta }{2(\ln \ln \beta)^{2k+1}}}{\ln \beta +1}
\;\ge\;
\frac{\ln \beta}{4(\ln \ln \beta)^{2k+1}}
\;\ge\;
\frac{\ln \beta}{4(\ln \ln \beta)^{\frac{\ln\ln  \beta}{\ln \ln \ln  \beta}-1}}
\\[1mm]
&=&
\frac{\ln \ln \beta}{4}\cdot
\frac{\ln\beta}{(\ln\ln\beta)^{\frac{\ln\ln\beta}{\ln\ln\ln\beta}}}
\;=\;
\frac{\ln \ln  \beta}{4}\cdot\frac{\ln\beta}{\ln\beta}
\;=\;
\frac{\ln\ln\beta}{4}\
\end{eqnarray*}

which means the algorithm is not constant competitive.
Therefore, for any algorithm to be constant competitive it must make at least $T$ choices with probability of more than $1-\frac{1}{\ln \ln \ln \beta }$ in the first  $n_{T}$ variables.
\end{proof}

To complete the proof of Theorem \ref{thm:lower bound} (in expectation) we use Claim \ref{clm:num of choices} (high probability) to get that any constant competitive algorithm must make at least
\[
\mathbb{E}[choices]\geq \left(1-\frac{1}{ \ln \ln \ln \beta}\right)\cdot \left(\frac{\ln \ln \beta}{2\ln \ln \ln \beta}-1\right)\geq \frac{\ln \ln \beta}{4\ln \ln \ln \beta} 
\]
choices in expectation. As mentioned at the beginning of the proof this yields the required bound of $\Omega \left(\frac{\ln \ln M}{\ln \ln \ln M}\right)$ for $M=O\left(\frac{e^n}{n}\right)$ given the relationship $M=\Theta\left( \frac{\beta}{\ln \beta}\right)$.
For $M=\Omega\left(\frac{e^n}{n}\right)$ we use $\beta=e^n$ and get a lower bound of $\Omega\left(\frac{\ln n}{\ln \ln n}\right)$.
\end{proof}

\section{A lower bound on the worst number of choices}
\label{sec:worstcase}
In the previous sections we provided a constant competitive algorithm with at most $O(\min\{\ln\ln M,\ln n\})$ choices in expectation. The natural question is to try to replace expectation with a deterministic upper bound on the number of choices. Unfortunately we show that this is impossible as a function of $n$. Specifically, we show that limiting the algorithm to at most $n\!-\!1$ choices prevents any constant competitive ratio, even for i.i.d.\ inputs. The proof constructs a discrete hard instance illustrating this limitation.

\begin{theorem}
Any algorithm for arbitrary i.i.d. variables that is always limited to $n-1$ choices cannot be constant competitive for large enough n.
\end{theorem}
\begin{proof}
Define the following probability distribution over a random variable \( X \).

Let
\[
 t = n^{n!}.
\]

Then \( X \) takes values as follows:

\[
\Pr[X = v_k] = 
\begin{cases}
\displaystyle \frac{1}{t},& \text{for } k=1,\\[2mm] 
\displaystyle \frac{n-1}{t},& \text{for } k=2,\\[2mm] 
\displaystyle \frac{n^{(k-1)!}-n^{(k-2)!}}{t}, & \text{for } 3 \leq k \leq n, \\[2mm]
1 - \displaystyle\frac{n^{(n-1)!}}{t}, & \text{for } k = n+1
\end{cases}
\]

where each value \( v_k \) is defined by:

\[
v_k = 
\begin{cases}
t^n, &\text{for } k=1,\\[2mm] 
\displaystyle \frac{t^n}{k \cdot \left( n^{(k-1)!}\right)^n}, & \text{for } 2 \leq k \leq n, \\[2mm]
0, & \text{for } k = n+1.
\end{cases}
\]

Specifically 
\( v_1 = t^n \), 
\( v_2 = \frac{t^n}{2 \cdot (n)^n} \), 
\( v_3 = \frac{t^n}{3 \cdot ( n^2)^n} \), 
and so on up to \( v_n \).

This defines a discrete probability distribution with \( n+1 \) outcomes.
We will show that any algorithm with $n-1$ choices doesn't achieve a constant competitive ratio.
\newline\textbf{Expectation of the Minimum.}
Let \(Y=\min\{x_1,\dots,x_n\}\).
\[
\mathbb{E}[Y]= \sum_{k=1}^{n+1} \Pr[Y=v_k]\cdot v_k \leq  \sum_{k=1}^{n+1}\Pr[Y\geq v_k]\cdot v_k  = \sum_{k=1}^{n}\Pr[Y\geq v_k]\cdot v_k\ .
\]
By definition for $k\ge 2$
\[
\Pr[Y\geq v_k] =  \frac{\left( 1+ (n-1)+ \sum_{j=2}^{k-1} \left(n^{j!}-n^{(j-1)!}\right)\right)^n}{t^n}=\frac{\left(n^{(k-1)!}\right)^n}{t^n} .
\]
This means that for $k\ge2$:
\[
\Pr[Y\geq v_k]\cdot v_k =  \frac{\left( n^{(k-1)!}\right)^n}{t^n}\cdot  \frac{t^n}{k \cdot \left( n^{(k-1)!}\right)^n}= \frac{1}{k}\ .
\]
Moreover, for $k=1$ we get that
\[
\Pr[Y\ge v_1] = \frac{1}{t^n} \ .
\]
Hence:
\[
\Pr[Y\ge v_1]\cdot v_1 = \frac{1}{t^n}\cdot t^n=1 \ .
\]

Ergo, in summary, we get that:
\[
\mathbb{E}[Y] \leq \sum_{k=1}^{n}\Pr[Y\geq v_k]\cdot v_k =\sum_{k=1}^{n}\frac{1}{k}\leq 2\ln (n) \ .
\]
\textbf{Expectation of cost of any algorithm with $n-1$ choices.}
We can analyze any algorithm by looking only at what it does on the following series:
\[x_1=v_2,x_2=v_3,...,x_{n-1}=v_n\] 
or a prefix of it.
Any algorithm when it sees the  variable $x_k=v_{k+1}$ must decide if to skip the variable $x_k$ and risk the option that all the next variables will be bigger than $x_k$. Otherwise, for every $1\le k\le n-1$ take the variable $x_{k}$ and risk the option that the last variable will be $v_{n+1}$.
Hence, we analyze only the following specific cases and see how each one contributes to the expectation of the algorithm:
\begin{align*}
\sigma_1 &= v_2,\{v_1\}^{n-1},\\
\sigma_2 &= v_2,v_3,\{v_1,v_2\}^{n-2},\\
\sigma_k &= v_2,v_3,...,v_k,v_{k+1},\{v_1,...,v_k\}^{n-k},\\
\sigma_{n-1} &= v_2,v_3,...,v_n,\{v_1,..,v_n,v_{n+1}\}^1,
\end{align*}
where $\{v_1,\dots,v_k\}^{n-k}$ denotes all sequences where the last $n-k$ variables are independently drawn from this set.

\begin{figure}[H]
  \centering
  \includegraphics[width=0.4\linewidth]{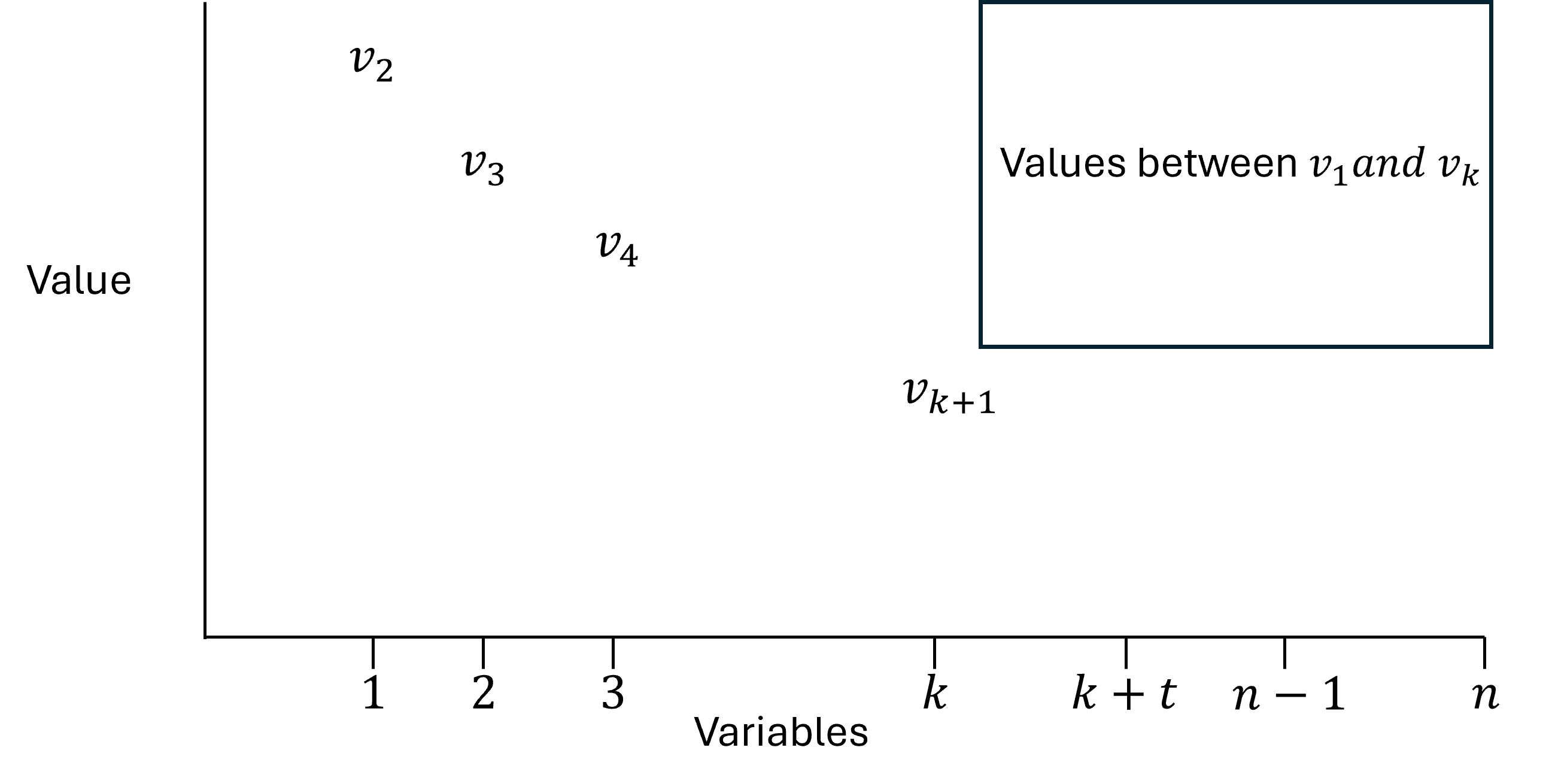}
  \caption{$\sigma_k$}
  \label{fig:sigma_k}
\end{figure}
Denote $ALG_k$ as the class of algorithms that skip $x_k$ when it encounters the series $x_1 =v_2,x_2=v_3 ,...x_k=v_{k+1}$. And denote $ALG_n$ the class of algorithms that don't skip any variable in the series  $x_1=v_2,x_2=v_3,...x_{n-1}=v_n$. Consider the specific prefix $x_1=v_2,\dots,x_{n-1}=v_n$.
Since the algorithm is deterministic, its behavior on this
prefix is fixed: at each step $i$ it either takes $x_i$ or
skips it. If there exists an index $k\in\{1,\dots,n-1\}$
such that the algorithm skips $x_k=v_{k+1}$ on this prefix,
let $k$ be the smallest such index; then the algorithm belongs
to the class $ALG_k$. Otherwise, the algorithm takes all of
$x_1,\dots,x_{n-1}$, so it belongs to the class $ALG_n$.
Thus, every deterministic algorithm belongs to at least one
of the classes $ALG_k$, $1\le k\le n$. It therefore suffices to
lower bound the expected cost of every algorithm in each class $ALG_k$.
For the class $ALG_n$ The only problematic series may be $\sigma_{n-1}$
\[
\Pr[\sigma = \sigma_{n-1}] = \Pr[x_1 =v_2,x_2=v_3,..,x_{n-1}=v_n] =  \frac{(n-1)\cdot\prod _{j=3}^{n}\left(n^{(j-1)!}-n^{(j-2)!}\right)}{t^{n-1}}\ .
\]
So the contribution to the expectation of $ALG_n$ is:

\begin{eqnarray*}
\Pr[\sigma=\sigma_{n-1}]\cdot v_n &=& \frac{(n-1)\cdot\prod _{j=3}^{n}\left(n^{(j-1)!}-n^{(j-2)!}\right)}{t^{n-1}}\cdot \frac{t^n}{n \cdot \left( n^{(n-1)!}\right)^n} \\
&=&  \frac{(n-1)\cdot\prod _{j=2}^{n-1}\left(n^{(j)!}-n^{(j-1)!}\right)}{n }\cdot \frac{t}{n^{n!}}
\geq \frac{n^{(n-1)!}}{4} \ge n    \ .
\end{eqnarray*}
where the first inequality uses 
\[
\frac{n-1}{n}\ge \frac{1}{2} \text{ and } \left(n^{(n-1)!}-n^{(n-2)!}\right)\ge\frac{1}{2}\cdot n^{(n-1)!},
\]
assuming $n\ge3$.

Hence, for every algorithm $A\in ALG_n$ we get that $\mathbb{E}[A]\geq n$ .
For the class $ALG_k$ the problematic series is $\sigma_k$ since it doesn't achieve the minimum value on it. The probability that this series happens is:
\[
\begin{aligned}
\Pr[\sigma=\sigma_k]
&=\left(\frac{(n-1)\cdot\prod_{j=3}^{k+1}\!\bigl(n^{(j-1)!}-n^{(j-2)!}\bigr)}{t^k}\right)
 \cdot\left(\frac{1+(n-1)+\sum_{j=3}^{k}\!\bigl(n^{(j-1)!}-n^{(j-2)!}\bigr)}{t}\right)^{n-k} \\
&=\left(\frac{(n-1)\cdot\prod_{j=3}^{k+1}\!\bigl(n^{(j-1)!}-n^{(j-2)!}\bigr)}{t^k}\right)\cdot\left(\frac{n^{(k-1)!}}{t}\right)^{n-k}\\
&\ge (n-1)\cdot\left(\prod_{j=2}^{k}\frac{1}{2}\cdot n^{(j)!}\right)\cdot\frac{ \left(n^{(k-1)!}\right)^{n-k}}{t^n}
\ge \frac{\left(\prod_{j=1}^{k} n^{(j)!}\right)}{2^k}\cdot \frac{ \left(n^{(k-1)!}\right)^{n-k}}{t^n}\ . 
\end{aligned}
\]

Note that the inequalities hold since $n^{(k-1)!}-n^{(k-2)!}\ge \frac{1}{2}\cdot n^{(k-1)!}$ for $k\ge 3$ and that $(n-1)\ge \frac{1}{2}\cdot n$  for $n\ge2$.
Hence,

\begin{eqnarray*}
\Pr[\sigma=\sigma_k]\cdot v_k & \ge & \frac{\prod_{j=1}^{k} n^{(j)!}}{2^k}\cdot \frac{ \left(n^{(k-1)!}\right)^{n-k}}{t^n}\cdot \frac{t^n}{k \cdot \left( n^{(k-1)!} \right)^n} \\
& = & \frac{\prod_{j=1}^k n^{(j)!}}{k \cdot 2^k\cdot n^{(k)!}} 
 = \frac{n^{(k)!}\prod_{j=1}^{k-1} n^{(j)!}}{k \cdot 2^k\cdot n^{(k)!} } \\
& =  & \frac{\prod_{j=1}^{k-1} n^{(j)!}}{k\cdot2^k} 
\geq  \frac{n}{8}\ .
\end{eqnarray*}

The last inequality holds for $k\geq 2$. Thus, for every $A\in ALG_k$ we get that $\mathbb{E}[A]\geq \frac{n}{8}$ for sufficiently large $n$ and for $k\ge2$.
For algorithm $A \in ALG_1$ the only problematic series we need
to consider is $\sigma_1$.

\begin{align*}
 \mathbb{E}[A] \geq \Pr[\sigma=\sigma_1]\cdot v_1 &= \frac{n-1}{t}\cdot \left(\frac{1}{t}\right)^{n-1}\cdot t^n= n-1 \ .
\end{align*}
Combining the above cases, we have shown that for every $k\in\{1,\dots,n\}$,
every deterministic algorithm $A$ in the class $ALG_k$ satisfies
$\mathbb{E}[A] \ge n/8$ for all sufficiently large $n$.
Since every deterministic algorithm belongs to $ALG_k$ for some
$k\in\{1,\dots,n\}$, it follows that every deterministic algorithm $A$
with at most $n-1$ choices satisfies
$\mathbb{E}[A]\ge n/8$.  Thus, the competitive ratio for any algorithm with $n-1$ choices is at least:
\[
\frac{\mathbb{E}[A]}{\mathbb{E}[Y]}\geq \frac{n}{16\ln(n)}
\]
which is unbounded as required, so no constant competitive ratio is possible.
\end{proof}

\section{Conclusion}
In this paper, we study the minimization prophet inequality under a multi-choice relaxation, where the algorithm is allowed to select multiple variables and pays the minimum among them. 
We show that for adversarial order an almost linear number of choices in expectation is required to achieve a constant-competitive ratio. 
However, for prophet secretary settings we prove that allowing $O(\min\{\ln \ln M, \ln n\})$ choices \textit{in expectation} is sufficient to achieve a constant competitive ratio (recall that $M$ is the ratio of the minimum expected value of any single variable to the expected minimum value of all variables).
We further establish that this bound is tight up to doubly or triply log factors. 
We also demonstrate that insisting on deterministic upper bound even for i.i.d.\ inputs, implies that $n$ choices are required to achieve a constant competitive ratio. Finally, we extend our algorithm to the multi-unit variant, showing that for i.i.d.\ inputs, a constant competitive ratio is achievable with only a small expected number of choices.
Our work leaves open the possibility that the impossibility results stem specifically from the magnitude of $M$, raising two open problems:

\begin{enumerate}
\item \textbf{Expected number of choices for adversarial order:} Our lower bound of $\Omega(\frac{n}{\ln n})$ expected choices for non-i.i.d.\ distributions uses a construction where $\ln \ln M \geq n$. Does there exist an algorithm for adversarial order that achieves a constant competitive ratio with $O(\min\{ \ln \ln M, \frac{n}{\ln n}\})$ choices in expectation?
\item \textbf{deterministic number of choices} Our lower bound showing that $n-1$ choices are insufficient relies on a distribution where $\ln \ln M \geq n$. Does there exist an algorithm for i.i.d.\ variables that achieves a constant competitive ratio with a deterministic upper bound of $O(\min\{ \ln \ln M, n\})$ on the number choices?
    
\end{enumerate}
\bibliography{references}

@article{KrengelSucheston1977,
  author    = {Ulrich Krengel and Louis Sucheston},
  title     = {Semiamarts and finite values},
  journal   = {Bulletin of the American Mathematical Society},
  volume    = {83},
  number    = {4},
  pages     = {745--747},
  year      = {1977},
  month     = jul,
  doi       = {10.1090/S0002-9904-1977-14375-0}
}

@incollection{KrengelSucheston1978,
  author    = {Ulrich Krengel and Louis Sucheston},
  title     = {On semiamarts, amarts, and processes with finite value},
  booktitle = {Probability on Banach Spaces},
  volume    = {4},
  pages     = {197--266},
  year      = {1978},
  publisher = {Dekker}
}

@article{HillK82,
  author    = {Hill, Theodore P. and Kertz, Robert P.},
  title     = {Stop rules and supremum expectations of i.i.d. random variables},
  journal   = {Annals of Probability},
  volume    = {10},
  number    = {2},
  pages     = {336--345},
  year      = {1982},
  publisher = {Institute of Mathematical Statistics},
  doi       = {10.1214/aop/1176993873}
}

@article{AssafSamuelCahn2000,
  author    = {Assaf, David and Samuel-Cahn, Ester},
  title     = {Simple ratio prophet inequalities for a mortal with multiple choices},
  journal   = {Journal of Applied Probability},
  volume    = {37},
  number    = {4},
  pages     = {1084--1091},
  year      = {2000},
  publisher = {Cambridge University Press},
  doi       = {10.1239/jap/1014843085}
}

@article{Assaf_Goldstein_Samuel-Cahn_2004, 
title={Two-choice optimal stopping},
volume={36}, 
DOI={10.1239/aap/1103662960},
number={4},
journal={Advances in Applied Probability},
author={Assaf, David and Goldstein, Larry and Samuel-Cahn, Ester},
year={2004},
pages={1116–1147}}

@article{BrustleEtAl2024,
  author       = {Johannes Brustle and
                  Jos{\'{e}} Correa and
                  Paul Duetting and
                  Victor Verdugo},
  title        = {The Competition Complexity of Dynamic Pricing},
  journal      = {Math. Oper. Res.},
  volume       = {49},
  number       = {3},
  pages        = {1986--2008},
  year         = {2024},
  url          = {https://doi.org/10.1287/moor.2022.0230},
  doi          = {10.1287/MOOR.2022.0230},
  bibsource    = {dblp computer science bibliography, https://dblp.org}
}

@article{EsfandiariEtAl2015,
  author       = {Hossein Esfandiari and
                  MohammadTaghi Hajiaghayi and
                  Vahid Liaghat and
                  Morteza Monemizadeh},
  title        = {Prophet Secretary},
  journal      = {{SIAM} J. Discret. Math.},
  volume       = {31},
  number       = {3},
  pages        = {1685--1701},
  year         = {2017},
  url          = {https://doi.org/10.1137/15M1029394},
  doi          = {10.1137/15M1029394},
  bibsource    = {dblp computer science bibliography, https://dblp.org}
}

@article{KERTZ198688,
title = {Stop rule and supremum expectations of i.i.d. random variables: A complete comparison by conjugate duality},
journal = {Journal of Multivariate Analysis},
volume = {19},
number = {1},
pages = {88-112},
year = {1986},
issn = {0047-259X},
doi = {https://doi.org/10.1016/0047-259X(86)90095-3},
url = {https://www.sciencedirect.com/science/article/pii/0047259X86900953},
author = {Robert P Kertz},
}

@inproceedings{LivanosMehta2022,
  author       = {Vasilis Livanos and
                  Ruta Mehta},
  editor       = {David P. Woodruff},
  title        = {Minimization is Harder in the Prophet World},
  booktitle    = {Proceedings of the 2024 {ACM-SIAM} Symposium on Discrete Algorithms,
                  {SODA} 2024, Alexandria, VA, USA, January 7-10, 2024},
  pages        = {424--461},
  publisher    = {{SIAM}},
  year         = {2024},
  url          = {https://doi.org/10.1137/1.9781611977912.17},
  doi          = {10.1137/1.9781611977912.17},
  bibsource    = {dblp computer science bibliography, https://dblp.org}
}

@article{Alaei2014,
  author       = {Saeed Alaei},
  title        = {Bayesian Combinatorial Auctions: Expanding Single Buyer Mechanisms
                  to Many Buyers},
  journal      = {{SIAM} J. Comput.},
  volume       = {43},
  number       = {2},
  pages        = {930--972},
  year         = {2014},
  url          = {https://doi.org/10.1137/120878422},
  doi          = {10.1137/120878422},
  bibsource    = {dblp computer science bibliography, https://dblp.org}
}

@inproceedings{Livanos2024,
  author       = {Vasilis Livanos and
                  Ruta Mehta},
  editor       = {Itai Ashlagi and
                  Aaron Roth},
  title        = {Minimization {I.I.D.} Prophet Inequality via Extreme Value Theory:
                  {A} Unified Approach},
  booktitle    = {Proceedings of the 26th {ACM} Conference on Economics and Computation,
                  {EC} 2025, Stanford University, Stanford, CA, USA, July 7-10, 2025},
  pages        = {1157--1179},
  publisher    = {{ACM}},
  year         = {2025},
  url          = {https://doi.org/10.1145/3736252.3742682},
  doi          = {10.1145/3736252.3742682},
  bibsource    = {dblp computer science bibliography, https://dblp.org}
}

@article{Eshani2024,
  author       = {Soheil Ehsani and
                  MohammadTaghi Hajiaghayi and
                  Thomas Kesselheim and
                  Sahil Singla},
  title        = {Prophet Secretary for Combinatorial Auctions and Matroids},
  journal      = {{SIAM} J. Comput.},
  volume       = {53},
  number       = {6},
  pages        = {1641--1662},
  year         = {2024},
  url          = {https://doi.org/10.1137/19m1264047},
  doi          = {10.1137/19M1264047},
  bibsource    = {dblp computer science bibliography, https://dblp.org}
}

@article{KLEINBERG2019,
  author       = {Robert Kleinberg and
                  S. Matthew Weinberg},
  title        = {Matroid prophet inequalities and applications to multi-dimensional
                  mechanism design},
  journal      = {Games Econ. Behav.},
  volume       = {113},
  pages        = {97--115},
  year         = {2019},
  url          = {https://doi.org/10.1016/j.geb.2014.11.002},
  doi          = {10.1016/J.GEB.2014.11.002},
  bibsource    = {dblp computer science bibliography, https://dblp.org}
}

@article{Chawala24,
  author       = {Shuchi Chawla and
                  Nikhil R. Devanur and
                  Thodoris Lykouris},
  title        = {Static Pricing for Multi-unit Prophet Inequalities},
  journal      = {Oper. Res.},
  volume       = {72},
  number       = {4},
  pages        = {1388--1399},
  year         = {2024},
  url          = {https://doi.org/10.1287/opre.2023.0031},
  doi          = {10.1287/OPRE.2023.0031},
  bibsource    = {dblp computer science bibliography, https://dblp.org}
}

@article{Jiashou22,
  author       = {Jiashuo Jiang and
                  Will Ma and
                  Jiawei Zhang},
  title        = {Tight Guarantees for Multiunit Prophet Inequalities and Online Stochastic
                  Knapsack},
  journal      = {Oper. Res.},
  volume       = {73},
  number       = {3},
  pages        = {1703--1721},
  year         = {2025},
  url          = {https://doi.org/10.1287/opre.2022.0309},
  doi          = {10.1287/OPRE.2022.0309},
  bibsource    = {dblp computer science bibliography, https://dblp.org}
}

@article{Correa21,
  author       = {Jos{\'{e}} Correa and
                  Patricio Foncea and
                  Ruben Hoeksma and
                  Tim Oosterwijk and
                  Tjark Vredeveld},
  title        = {Posted Price Mechanisms and Optimal Threshold Strategies for Random
                  Arrivals},
  journal      = {Math. Oper. Res.},
  volume       = {46},
  number       = {4},
  pages        = {1452--1478},
  year         = {2021},
  url          = {https://doi.org/10.1287/moor.2020.1105},
  doi          = {10.1287/MOOR.2020.1105},
  bibsource    = {dblp computer science bibliography, https://dblp.org}
}

@article{Junjie24,
  author       = {Junjie Qin and
                  Shai Vardi and
                  Adam Wierman},
  title        = {Minimization Fractional Prophet Inequalities for Sequential Procurement},
  journal      = {Math. Oper. Res.},
  volume       = {49},
  number       = {2},
  pages        = {928--947},
  year         = {2024},
  url          = {https://doi.org/10.1287/moor.2021.0173},
  doi          = {10.1287/MOOR.2021.0173},
  bibsource    = {dblp computer science bibliography, https://dblp.org}
}

@inproceedings{Abolhassani17,
  author       = {Melika Abolhassani and
                  Soheil Ehsani and
                  Hossein Esfandiari and
                  MohammadTaghi Hajiaghayi and
                  Robert D. Kleinberg and
                  Brendan Lucier},
  editor       = {Hamed Hatami and
                  Pierre McKenzie and
                  Valerie King},
  title        = {Beating 1-1/e for ordered prophets},
  booktitle    = {Proceedings of the 49th Annual {ACM} {SIGACT} Symposium on Theory
                  of Computing, {STOC} 2017, Montreal, QC, Canada, June 19-23, 2017},
  pages        = {61--71},
  publisher    = {{ACM}},
  year         = {2017},
  url          = {https://doi.org/10.1145/3055399.3055479},
  doi          = {10.1145/3055399.3055479},
  bibsource    = {dblp computer science bibliography, https://dblp.org}
}

@inproceedings{Azar18,
  author       = {Yossi Azar and
                  Ashish Chiplunkar and
                  Haim Kaplan},
  editor       = {{\'{E}}va Tardos and
                  Edith Elkind and
                  Rakesh Vohra},
  title        = {Prophet Secretary: Surpassing the 1-1/e Barrier},
  booktitle    = {Proceedings of the 2018 {ACM} Conference on Economics and Computation,
                  Ithaca, NY, USA, June 18-22, 2018},
  pages        = {303--318},
  publisher    = {{ACM}},
  year         = {2018},
  url          = {https://doi.org/10.1145/3219166.3219182},
  doi          = {10.1145/3219166.3219182},
  bibsource    = {dblp computer science bibliography, https://dblp.org}
}

@article{Arnosti23,
  author       = {Nick Arnosti and
                  Will Ma},
  title        = {Tight Guarantees for Static Threshold Policies in the Prophet Secretary
                  Problem},
  journal      = {Oper. Res.},
  volume       = {71},
  number       = {5},
  pages        = {1777--1788},
  year         = {2023},
  url          = {https://doi.org/10.1287/opre.2022.2419},
  doi          = {10.1287/OPRE.2022.2419},
  bibsource    = {dblp computer science bibliography, https://dblp.org}
}

@article{Marek20,
  author       = {Marek Adamczyk and
                  Michal Wlodarczyk},
  title        = {Multi-dimensional mechanism design via random order contention resolution
                  schemes},
  journal      = {SIGecom Exch.},
  volume       = {17},
  number       = {2},
  pages        = {46--53},
  year         = {2019},
  url          = {https://doi.org/10.1145/3381329.3381334},
  doi          = {10.1145/3381329.3381334},
  bibsource    = {dblp computer science bibliography, https://dblp.org}
}

@inproceedings{Ezra18,
  author       = {Tomer Ezra and
                  Michal Feldman and
                  Ilan Nehama},
  editor       = {{\'{E}}va Tardos and
                  Edith Elkind and
                  Rakesh Vohra},
  title        = {Prophets and Secretaries with Overbooking},
  booktitle    = {Proceedings of the 2018 {ACM} Conference on Economics and Computation,
                  Ithaca, NY, USA, June 18-22, 2018},
  pages        = {319--320},
  publisher    = {{ACM}},
  year         = {2018},
  url          = {https://doi.org/10.1145/3219166.3219211},
  doi          = {10.1145/3219166.3219211},
  bibsource    = {dblp computer science bibliography, https://dblp.org}
}

@inproceedings{EkbataniNNV24,
  author       = {Farbod Ekbatani and
                  Rad Niazadeh and
                  Pranav Nuti and
                  Jan Vondr{\'{a}}k},
  editor       = {Bojan Mohar and
                  Igor Shinkar and
                  Ryan O'Donnell},
  title        = {Prophet Inequalities with Cancellation Costs},
  booktitle    = {Proceedings of the 56th Annual {ACM} Symposium on Theory of Computing,
                  {STOC} 2024, Vancouver, BC, Canada, June 24-28, 2024},
  pages        = {1247--1258},
  publisher    = {{ACM}},
  year         = {2024},
  url          = {https://doi.org/10.1145/3618260.3649786},
  doi          = {10.1145/3618260.3649786},
  bibsource    = {dblp computer science bibliography, https://dblp.org}
}

@inproceedings{LT25,
  author       = {Ziyun Chen and
                  Zhiyi Huang and
                  Dongchen Li and
                  Zhihao Gavin Tang},
  editor       = {Yossi Azar and
                  Debmalya Panigrahi},
  title        = {Prophet Secretary and Matching: the Significance of the Largest Item},
  booktitle    = {Proceedings of the 2025 Annual {ACM-SIAM} Symposium on Discrete Algorithms,
                  {SODA} 2025, New Orleans, LA, USA, January 12-15, 2025},
  pages        = {1371--1401},
  publisher    = {{SIAM}},
  year         = {2025},
  url          = {https://doi.org/10.1137/1.9781611978322.42},
  doi          = {10.1137/1.9781611978322.42},
  bibsource    = {dblp computer science bibliography, https://dblp.org}
}

@article{FredM23,
  author       = {Giordano Giambartolomei and
                  Frederik Mallmann{-}Trenn and
                  Raimundo Saona},
  title        = {Prophet Inequalities: Separating Random Order from Order Selection},
  journal      = {CoRR},
  volume       = {abs/2304.04024},
  year         = {2023},
  url          = {https://doi.org/10.48550/arXiv.2304.04024},
  doi          = {10.48550/ARXIV.2304.04024},
  eprinttype    = {arXiv},
  eprint       = {2304.04024},
  bibsource    = {dblp computer science bibliography, https://dblp.org}
}

@inproceedings{Cor2017,
author = {Correa, Jos\'{e} and Foncea, Patricio and Hoeksma, Ruben and Oosterwijk, Tim and Vredeveld, Tjark},
title = {Posted Price Mechanisms for a Random Stream of Customers},
year = {2017},
isbn = {9781450345279},
publisher = {Association for Computing Machinery},
address = {New York, NY, USA},
url = {https://doi.org/10.1145/3033274.3085137},
doi = {10.1145/3033274.3085137},
booktitle = {Proceedings of the 2017 ACM Conference on Economics and Computation},
pages = {169–186},
numpages = {18},
location = {Cambridge, Massachusetts, USA},
series = {EC '17}
}

@article{SamuelCahn84,
 ISSN = {00911798, 2168894X},
 URL = {http://www.jstor.org/stable/2243359},
 author = {Ester Samuel-Cahn},
 journal = {The Annals of Probability},
 number = {4},
 pages = {1213--1216},
 publisher = {Institute of Mathematical Statistics},
 title = {Comparison of Threshold Stop Rules and Maximum for Independent Nonnegative Random Variables},
 urldate = {2026-02-08},
 volume = {12},
 year = {1984}
}

@inproceedings{Harb25,
  author       = {Elfarouk Harb},
  editor       = {Yossi Azar and
                  Debmalya Panigrahi},
  title        = {New Prophet Inequalities via Poissonization and Sharding},
  booktitle    = {Proceedings of the 2025 Annual {ACM-SIAM} Symposium on Discrete Algorithms,
                  {SODA} 2025, New Orleans, LA, USA, January 12-15, 2025},
  pages        = {1222--1269},
  publisher    = {{SIAM}},
  year         = {2025},
  url          = {https://doi.org/10.1137/1.9781611978322.37},
  doi          = {10.1137/1.9781611978322.37},
  bibsource    = {dblp computer science bibliography, https://dblp.org}
}

@inproceedings{HarbPeled24,
  author       = {Sariel Har{-}Peled and
                  Elfarouk Harb and
                  Vasilis Livanos},
  editor       = {Karl Bringmann and
                  Martin Grohe and
                  Gabriele Puppis and
                  Ola Svensson},
  title        = {Oracle-Augmented Prophet Inequalities},
  booktitle    = {51st International Colloquium on Automata, Languages, and Programming,
                  {ICALP} 2024, Tallinn, Estonia, July 8-12, 2024},
  series       = {LIPIcs},
  pages        = {81:1--81:19},
  publisher    = {Schloss Dagstuhl - Leibniz-Zentrum f{\"{u}}r Informatik},
  year         = {2024},
  url          = {https://doi.org/10.4230/LIPIcs.ICALP.2024.81},
  doi          = {10.4230/LIPICS.ICALP.2024.81},
  bibsource    = {dblp computer science bibliography, https://dblp.org}
}

\section{Missing proofs from section \ref{sec:non_iid}}
\label{app: non_iid}
\subsection{Proof of Claim~\ref{claim:expectation non i.i.d}}
\label{prf: expectation non i.i.d}
\begin{proof}

Let
\[
Y \;=\; \min\{X_1,X_2,\dots,X_n\}. 
\]
Now we will calculate the CDF of Y:
\begin{align*}
\Pr[Y>x] &=\Pr[\min\{X_1,\dots,X_n\}>x] = \prod_{k=1}^{n}\Pr[X_k>x] \\ 
&=\left(\prod_{k=1}^{n-1}\frac{1}{x^{{n^{2k-2n}}}}\right)\cdot \dfrac{1}{x^{1-{\sum_{k=1}^{n-1}n^{2k-2n}}}}=\frac{1}{x}\ .
\end{align*}
Hence,
\[
F_Y(x) =\begin{cases}
1 - \dfrac{1}{x}, & 1\le x < e^{n^{2n}}, \\[1mm]
1, & x\ge e^{n^{2n}}.
\end{cases}\ .
\]
Thus, the expectation of \(Y\) is given by
\[
\mathbb{E}[Y]=1 + \int_{1}^{e^{n^{2n}}} \,1- F_Y(x)\,dx = 1 + \int_{1}^{e^{n^{2n}}}\frac{1}{x}
= 1+ \ln (e^{n^{2n}}) -\ln(1) = 1+n^{2n}.
\]

Hence,

\[
\mathbb{E}\Bigl[\min\{X_1,\dots,X_n\}\Bigr] = 1+n^{2n}. 
\]   
\end{proof}

\subsection{Proof of Claim~\ref{clm: X_k values}}
\label{prf: X_k values}
\begin{proof}

We define the CDF of $Y_k$ as $F_{Y_k}$.
\begin{align*}
\Pr[Y_k>x] &=\Pr[\min\{X_1,\dots,X_k\}>x] = \prod_{i=1}^{k}\Pr[X_i>x] \\ 
&=\left(\prod_{i=1}^{k}\frac{1}{x^{{n^{2i-2n}}}}\right)=\frac{1}{x^{\sum_{i=1}^k n^{2i-2n}}}=\frac{1}{x^{\frac{\sum_{i=1}^k n^{2i}}{n^{2n}}}}.
\end{align*}

Hence, $F_{Y_k}=1-{x^{-\frac{\sum_{i=1}^k n^{2i}}{n^{2n}}}}$. By the definition of the CDF

\[
\Pr\!\bigl[A_k\le Y_k]
 =1-F_{Y_k}(A_k) \ .
\]

First we bound $F_{Y_k}(A_k)$
\begin{align*}
F_{Y_k}(A_k)
&=1-A_k^{-\frac{\sum_{i=1}^k n^{2i}}{n^{2n}}}
  =1-e^{-n^{2n-2k-1}\cdot\frac{\sum_{i=1}^k n^{2i}}{n^{2n}}}
  =1-e^{-\frac{\sum_{i=1}^{k}n^{2i}}{n^{2k+1}}} \\
&\le 1-e^{-\frac{2n^{2k}}{n^{2k+1}}}
  =1-e^{-\frac{2}{n}}
  \le \frac{2}{n},
\end{align*}

where the first inequality holds for $n\ge 2$ since $n^{2k}$ dominates the sum, and the last inequality holds since $e^{-a}\ge 1-a$.
Hence, we get that for every $k$ :
$\Pr[Y_k\ge A_k]  = 1-F_{Y_k}(A_k) \ge 1-\frac{2}{n}.$\
\end{proof}
For each $k$, define
\[
Z_k = \min\{X_{k+1},\dots,X_n\},
\qquad
C_k = \mathbb{E}\!\left[\min\!\left(A_{k-1},Z_k\right)\right].
\]

\section{Missing proofs from section \ref{sec:lowerbound_expected}}
\label{app: lowerbound_expected}
\subsection{Proof of Claim ~\ref{clm: expectation min i.i.d}}
\label{prf:expectation min i.i.d}
\begin{proof}
Let \(X_1, X_2, \dots, X_n\) be independent copies of \(X\) (with the above CDF). Define
\[
Y = \min\{X_1, X_2, \dots, X_n\}\,.
\]

Then, for \(1 \le x < \beta\), we have:
\[
P(Y > x) = (1 - F(x))^n = \left( \frac{1}{x^{1/n}} \right)^n = \frac{1}{x}\,,
\]
so the CDF of \(Y\) is
\[
F_Y(x) = \begin{cases}
1 - \dfrac{1}{x}, & 1\le x < \beta, \\[1mm]
1, & x\ge \beta,
\end{cases}.
\]

At the point \(x = \beta\), the mass is:
\[
P(Y = \beta) = \left(P(X = \beta)\right)^n = \left(\frac{1}{\beta^{1/n}}\right)^n = \frac{1}{\beta}.
\]

The expectation of \(Y\) is given by:
\[
\mathbb{E}[Y] = \int_1^{\beta} x \cdot f_Y(x)\,dx + \beta \cdot \frac{1}{\beta}.
\]

We compute the derivative \(f_Y(x) = \frac{d}{dx} F_Y(x) = \frac{1}{x^2}\), so:
\[
\mathbb{E}[Y] = \int_1^{\beta} x \cdot \frac{1}{x^2}\,dx +\beta \cdot \frac{1}{\beta}
= \int_1^{\beta} \frac{1}{x}\,dx + 1 = \ln \beta + 1 .
\]
In other words, 
\[
\boxed{\mathbb{E}\Bigl[\min\{X_1, \dots, X_n\}\Bigr] = \ln \beta + 1}
\]
\end{proof}

\section{The multi-unit algorithm}
\label{sec:multiunit_algorithm}
In this section we present an algorithm almost identical to that of Section~\ref{sec:algorithm}, which achieves a constant competitive ratio in the multi-unit minimization problem while using only a small expected number of choices. We analyze both its competitive ratio and its expected number of choices.

We define the following algorithm:

\begin{center}
\fbox{\parbox{0.93\textwidth}{
\textbf{Algorithm 2}
\begin{enumerate}
    \item  Find a threshold $t$ (e.g., by binary search) such that 
    \[
    \sum_i p_i = \min\{\max\{4\ln M+r,\,5r\},\,n\},
    \]
    where $p_i = \Pr(X_i\le t)$.
    \item Wait until the first $r$ items with values at most $t$ appear.  
    If only $\ell < r$ such items appear, take the remaining $r-\ell$ items from the last arrivals.
    \item Thereafter, switch only if a newly observed value is strictly smaller than the highest value currently held among the $r$ items.
\end{enumerate}
}}
\end{center}

\begin{theorem}
For i.i.d.\ variables and $r \le \ln M$, the algorithm above is constant competitive  and uses 
\[
O\!\left(\min\!\left\{r\ln\!\tfrac{\ln M}{r},\;r\ln\!\tfrac{n}{r}\right\}\right)
\]
choices in expectation.
\end{theorem}

\begin{theorem}
For i.i.d.\ variables and $r \ge \ln M$, the algorithm above is constant competitive and uses $ O\!\left(r\right)$ choices in expectation.
\end{theorem}

\subsection{Analysis}
For the analysis we assume $\mathbb{E}[\min_jX_j]=1$ which means that $\min _j\mathbb{E}[X_j]=M$.
We analyze the algorithm separately for its competitive ratio and its expected number of choices.  
The argument closely follows that of Section~\ref{sec:algorithm}, except that the algorithm now maintains up to $r$ values simultaneously and switches only when a newly observed value is smaller than the highest value currently held.  
The threshold $t$ is chosen so that
\[
\sum_i p_i = \min\{\max\{4\ln M+r,\,5r\},\,n\}.
\]

\subsubsection{Competitive ratio analysis}

The reasoning is analogous to the single-unit case.  
The algorithm may fail to achieve the optimal $r$-minimum only if from the $n-r$ first values less than  $r$ values fall below the threshold~$t$.  We will denote $w_i$ as the indicator s.t. $x_i\le t$ and $S=\sum_{i=1}^{n-r}w_i$.

\paragraph*{Case 1: $r\le \ln M$.}
Here $\sum_{i=1}^n p_i=4\ln M+r$.
And since for each variable $x_i$ the maximum $p_i$ can be is 1. We get that $\sum_{i=1}^{n-r}p_i\geq 4 \ln M$
We wish to bound the probability that from the first $n-r$ variables, fewer than $r$ are below the threshold $t$. Hence, $\mathbb{E}[S]\ge 4\ln M$ and $r\le (1-\delta)\mathbb{E}[S]$ for $\delta=1-\frac{r}{\mathbb{E}[S]}$.  
By the Chernoff bound,
\[
\Pr(S<r) \le \exp\!\left(-\frac{\delta^2\,\mathbb{E}[S]}{2}\right)
\le \exp\!\left(-\frac{(1-\frac{r}{4\ln M})^2\,4\ln M}{2}\right)
\le e^{-2\ln M + r} \le e^{-\ln M} = \frac{1}{M}.
\]
Thus, with probability at least $1-1/M$, the algorithm observes at least $r$ samples below the threshold from the first $n-r$ variables,  and obtains the true $r$-minimum.  
The failure event contributes at most $rM\cdot(1/M)=r$ to the expected cost.

\paragraph*{Case 2: $r\ge \ln M$.}
By the same reasoning  $\mathbb{E}[S] \ge4r$ and  taking $\delta=3/4$.  
By the Chernoff bound we get:
\[
\Pr(S<r)\le \exp\!\left(-\frac{\delta^2\,\mathbb{E}[S]}{2}\right)
\le e^{-(3/4)^2 2r}=e^{-9r/8}\le e^{-r}\le \frac{1}{M}.
\]
Again, the failure event contributes at most $r$ to the expected cost.

\paragraph*{Case 3: $\sum_{i=1}^np_i = n$}. Here, the threshold $t$ equals $\infty$. Hence, all samples are relevant, and the algorithm always achieves the minimum, which means it is 1 competitive.
Hence, the following holds for every r:
\[
\mathbb{E}[\text{ALG}]
\;\le\; \mathbb{E}[\text{OPT}] + rM\cdot\frac{1}{M}
\;\le\; 2\,\mathbb{E}[\text{OPT}],
\]
so the algorithm is 2-competitive.  
Since the expectation of the minimum value is 1 than $\mathbb{E}[\text{OPT}] \ge r$ (the expectation of the $r$ smallest values).

\subsubsection{Expected number of choices}

We extend the single-unit record-minima argument to the $r$-unit case.  
Initially, the algorithm accepts the first $r$ items.  
Each subsequent item replaces the current maximum among the $r$ held items if it is smaller.  
Hence, the expected number of switches equals the expected number of times a new element enters the current $r$-minimum set:
\[
\mathbb{E}[\text{Choices}] 
= r + \sum_{i=r+1}^{n} \frac{r}{i}
= r + r(H_n - H_r)
= \Theta\!\left(r\ln\!\tfrac{n}{r}\right),
\]
where $H_m$ denotes the $m$-th harmonic number.

We now distinguish three parameter regimes according to $\sum_i p_i$:
\begin{enumerate}
    \item $\sum_i p_i = 4\ln M+r$.  
          The analysis mirrors Section~\ref{sec:algorithm}; scaling by~4 does not affect asymptotics.  
          Since the algorithm tracks $r$ values, the expected number of choices is $O\!\left(r\ln\frac{\ln M}{r}+\ln r\right)=O\!\left(r\ln\frac{\ln M}{r}\right)$.
    \item $\sum_i p_i = 5r$.  
      By the record-minima argument, the expected number of choices among $O(r)$ samples is 
      $r + r(H_{4r}-H_r) = O(r)$.  
     
    \item $\sum_i p_i = n$.  
          This regime occurs when all items are relevant (i.e., all below the threshold).  
          In that case, every incoming item may potentially replace one of the currently held $r$ items, exactly following the record-minima process generalized to the $r$-unit case.  
          As derived above,
          \[
          \mathbb{E}[\text{Choices}] = r + r(H_n - H_r) = \Theta\!\left(r\ln\frac{n}{r}\right).
          \]
\end{enumerate}
Hence, when $r\leq \ln M$, the algorithm achieves a constant competitive ratio while using
\[
O\!\left(\min\!\left\{r\ln\!\tfrac{\ln M}{r},\;r\ln\!\tfrac{n}{r}\right\}\right)
\]
choices in expectation.  
When $r\ge \ln M$, our threshold choice enforces $\sum_i p_i=\min\{4r,n\}$, so the expected number of relevant samples is at most $4r$; consequently, the algorithm uses
$O(r)$ choices in expectation.

\section{Lower bound for the multi-unit }
\label{sec:multiunit_lowerbound}
We now consider the case where the algorithm selects exactly \(r\) items and its performance is measured by the sum of their values. We show that if the algorithm makes exactly $r$ choices, then the competitive ratio in this case grows exponentially with \(n\).

We consider the minimization variant with exactly $r$ selections. Let $X_1,\dots,X_n$ be i.i.d.\ where each variable gets 3 possible values. Specifically, 
\[
\Pr[X_i=1]=1-\tfrac{1}{3^n}-\tfrac{1}{9^n},\qquad
\Pr[X_i=A]=\tfrac{1}{3^n},\qquad
\Pr[X_i=B]=\tfrac{1}{9^n},
\]
where
\[
A=3^{\,n^2-n(r-1)}=3^{\,n^2-nr+n},\qquad
B=9^{\,n^2-n(r-1)}=3^{\,2n^2-2nr+2n}.
\]
The prophet benchmark is the sum of the $r$ \emph{smallest} order statistics:
\[
\mathbb{E}[OPT(r)]=\mathbb{E}\Big[\sum_{j=1}^r X_{(j)}\Big].
\]

\begin{theorem}
For every $1\le r\le n-1$, any online algorithm that must select exactly $r$ items satisfies
\[
\frac{\mathbb{E}[ALG(r)]}{\mathbb{E}[OPT(r)]} \;\ge\; \frac{1}{n^3}\left(\frac{3}{2}\right)^n \ .
\]
In particular, the competitive ratio is exponential in $n$.

\end{theorem}
\begin{proof}
By Yao's principle, it suffices to lower bound the expected cost of any deterministic online
algorithm under the i.i.d.\ distribution above.

\noindent\paragraph*{Upper bound on $\mathbb{E}[OPT(r)]$.}
Denote $p_A=3^{-n}$ and $p_B=9^{-n}$. Let $N_A$ (resp.\ $N_B$) be the number of $A$’s (resp.\ $B$’s).
If at least $r$ ones appear, the prophet pays exactly $r$. Otherwise, among the $r$ smallest,
the prophet includes some non-ones. We will separate the possible contributions by value type.

\paragraph*{Contribution of $A$.} If the prophet includes any $A$, then necessarily $N_A\ge 1$ and fewer than
$r$ ones appear. In particular, the event $\{N_A+N_B\ge n-r+1\}$ is necessary to force
at least one $A$ into the $r$ smallest. In total, there are less than $n^2$ options for $N_A,N_B$ and the highest probability that the prophet includes $A$ is when $N_A=n-r+1$ and $N_B= 0$. This happens with probability of less than:
\[
\binom{n}{n-r+1}p_A^{n-r+1} \leq 2^n p_A^{n-r+1}\ .
\]
Since there are at most $n^2$ choices for $N_A,N_B$ we can upper bound the probability of the occurrence that the prophet must take at least one $A$ by $n^22^np_A^{n-r+1}$. Note that for these cases OPT is at most $r3^{\,n^2-nr+n}$ so we can upper bound the expectation of OPT if it selects $A$ and not $B$ by:
\[
n^22^np_A^{n-r+1}\cdot r3^{\,n^2-nr+n}\leq n^32^n3^{-n^2+nr-n}\cdot 3^{n^2-nr+n} = n^32^n \ .
\]

\paragraph*{Contribution of $B$.} For the prophet to take $B$ it must be that $N_B\geq n-r+1$. This happens with probability of less than:
\[
\binom{n}{n-r+1}p_B^{n-r+1} \leq 2^n p_B^{n-r+1} \ .
\]
Since there are at most $n$ choices for $N_B$ we can upper bound the probability of the occurrence that the prophet must take at least one $B$ by $2^n p_B^{n-r+1}$. Note that for these cases OPT is at most $r9^{n^2-nr+n}$ so we can upper bound the expectation of OPT if it selects $B$ by:
\[
n2^n p_B^{n-r+1}\cdot r9^{n^2-nr+n}\leq n^2 2^n 9^{-n^2-nr+n}\cdot9^{n^2+nr-n}=n^22^n \ .
\]

Combining and adding the event where at least $r$ ones appear,
\[ 
\mathbb{E}[OPT(r)]\;\le\; n \;+\; n^3\cdot 2^n \;+\; n^2\cdot 2^n \;\le\; 3n^3\cdot 2^ n
\]

\noindent\paragraph*{Lower bound on $\mathbb{E}[ALG]$.}
Let $\mathcal{F}$ be the event that the first $n-r$ arrivals are all $A$:
\[
\Pr[\mathcal{F}] \;=\; p_A^{\,n-r} \;=\; (3^{-n})^{\,n-r} \;=\; 3^{-n(n-r)}.
\]
Fix any deterministic online policy and condition on $\mathcal{F}$. After observing the prefix of
$n-r$ many $A$'s, the policy is in one of two states:

\emph{(i) It has accepted at least one $A$ in the prefix.}
Then it has already incurred cost at least $A$.

\emph{(ii) It has accepted none in the prefix.}
It must accept all of the last $r$ arrivals to reach exactly $r$ picks. Among these last $r$
positions, a $B$ appears with probability
\[
\Pr[\text{at least one $B$ in last $r$}] \;=\; 1-(1-p_B)^r \;\ge\; r\,p_B \;=\; r\cdot 9^{-n}.
\]
On this sub-event, the policy is forced to include a $B$, paying at least $B$. Therefore its
conditional expected cost is at least $r\,p_B\,B = r\cdot 9^{-n}\cdot 9^{n^2-nr+n}
= r\cdot 9^{\,n^2-nr}$.

Comparing the two lower bounds under $\mathcal{F}$, note that for sufficiently large $n$
\[
A \;=\; 3^{\,n^2-nr+n}
\;\;\le\;\; r\cdot 9^{\,n^2-nr}
\]
so in either case the conditional expected cost given $\mathcal{F}$ is at least $A$.
So the expectation this event adds to the total expectation is more than:
\[
\Pr[\mathcal{F}]\cdot 3^{n^2-nr+n} =3^{-n^2+nr}\cdot 3^{n^2-nr+n} = 3^n \ .
\]
Hence, we get that $\mathbb{E}[ALG]\geq 3^n$ so the competitive ratio is at least:
\[
\frac{\mathbb{E}[ALG(r)]}{\mathbb{E}[OPT(r)]}\geq \frac{3^n}{n^32^n} =\frac{1}{n^3}\left(\frac{3}{2}\right)^n
\]
which is exponential in $n$ which completes the proof.
\end{proof}

\section{Examples of M}
\label{App: M examples}
In this section, we compute the parameter
\[
M \;=\; \frac{\min_j \mathbb{E}[X_j]}{\mathbb{E}[\min_j X_j]}
\]
for two simple examples.

\paragraph*{Uniform distribution $U[0,1]$.}
Let $X_1,\dots,X_n$ be i.i.d.\ with $X_j \sim U[0,1]$.
Then $\mathbb{E}[X_j] = 1/2$ for all $j$, so
$\min_j \mathbb{E}[X_j] = \frac{1}{2}$.
The minimum of $n$ i.i.d.\ uniform $[0,1]$ variables has expectation
$\mathbb{E}[\min_j X_j] = 1/(n+1)$. Hence
\[
M
= \frac{1/2}{1/(n+1)}
= \frac{n+1}{2}
= \Theta(n).
\]
Thus for a ``nice'' bounded distribution such as $U[0,1]$, the difficulty
parameter $M$ grows linearly with $n$.

\paragraph*{A truncated heavy-tailed distribution.}
Fix $Z \le e^n$ and let $X_1,\dots,X_n$ be i.i.d.\ with distribution $D_Z$
whose CDF is
\[
F(x)=
\begin{cases}
1 - x^{-1/n}, & 1 \le x < Z, \\[1mm]
1, & x \ge Z.
\end{cases}
\]
For $1 \le x < Z$ the density is
\[
f(x) = F'(x) = \frac{1}{n} x^{-1 - 1/n},
\]
and there is a point mass at $x=Z$ of size
\[
\mathbb{P}(X_j = Z)
= 1 - \lim_{x \to Z^-} F(x)
= \frac{1}{Z^{1/n}}.
\]

\emph{Expectation of a single variable.}
We have
\[
\mathbb{E}[X_j]
= \int_1^Z x f(x)\,dx + Z \cdot \mathbb{P}(X_j = Z)
= \int_1^Z x \cdot \frac{1}{n} x^{-1 - 1/n} \, dx
  + Z \cdot \frac{1}{Z^{1/n}}.
\]
The integral term simplifies to
\[
\int_1^Z \frac{1}{n} x^{-1/n} \, dx
= \frac{1}{n} \cdot \frac{n}{n-1} \bigl(Z^{1-1/n} - 1\bigr)
= \frac{Z^{1-1/n} - 1}{n-1}.
\]
Adding the atom at $Z$ gives
\[
\mathbb{E}[X_j]
= \frac{Z^{1-1/n} - 1}{n-1} + Z^{1-1/n}
= \frac{Z n - Z^{1/n}}{Z^{1/n}(n-1)}
= \Theta\!\bigl(Z^{1-1/n}\bigr)
\quad\text{for fixed $n$ and large $Z$}.
\]
Since all variables have the same distribution, $\min_j \mathbb{E}[X_j]
= \mathbb{E}[X_j] = \Theta(Z^{1-1/n})$.

\emph{Expectation of the minimum.}
Let $Y = \min_j X_j$. For $1 \le t < Z$ we have
\[
\mathbb{P}(X_j \ge t)
= 1 - F(t)
= \frac{1}{t^{1/n}},
\]
so
\[
\mathbb{P}(Y \ge t)
= \mathbb{P}(X_1 \ge t,\dots,X_n \ge t)
= \left(\frac{1}{t^{1/n}}\right)^n
= \frac{1}{t}.
\]
For $0 \le t < 1$ we have $\mathbb{P}(Y \ge t) = 1$, and for $t \ge Z$,
$\mathbb{P}(Y \ge t) = 0$. Using the standard identity
$\mathbb{E}[Y] = \int_0^\infty \mathbb{P}(Y \ge t)\,dt$ for nonnegative
random variables, we obtain
\[
\mathbb{E}[Y]
= \int_0^1 1 \, dt
+ \int_1^Z \frac{1}{t} \, dt
= 1 + \ln Z.
\]

Combining the above,
\[
M
= \frac{\min_j \mathbb{E}[X_j]}{\mathbb{E}[\min_j X_j]}
= \Theta\!\left(\frac{Z^{1-1/n}}{1 + \ln Z}\right).
\]
Thus, in this example, the ratio between the expected value of a single
variable and the expected minimum can be huge. In particular, if
$Z = e^n$ we obtain
\[
M = \Theta\!\left(\frac{e^{n}}{n}\right),
\]
so $M$ can be exponential in $n$ even though $Z$ is only $e^n$.
\end{document}